\pdfoutput=1

\documentclass[english]{article}

\usepackage[T1]{fontenc}
\usepackage[utf8]{inputenc} 
\usepackage{lmodern}
\usepackage{microtype}
\usepackage[english]{babel}

\usepackage[left=0.52in,right=0.52in,top=0.4in,bottom=0.5in,includeheadfoot]{geometry}

\usepackage{graphicx}
\usepackage{float}
\usepackage{booktabs}

\usepackage{amsmath,amssymb}
\usepackage{amsthm}
\newtheorem{theorem}{Theorem}
\theoremstyle{remark}
\newtheorem{remark}{Remark}

\usepackage[numbers,sort&compress]{natbib}

\usepackage{animate}

\usepackage[hidelinks]{hyperref}

\makeatother
\title{Models with Accelerated Failure Conditionals%
\thanks{For full functionality of the animated figures presented in this document, please view this PDF in \textbf{Adobe Acrobat Reader}. Other PDF viewers (including browser-based viewers, 
Preview) may render the figures statically.}}
\author{
  Jared N. Lakhani\\
  \textit{Department of Statistical Sciences, University of Cape Town}\\
  \texttt{lkhjar001@myuct.ac.za}
}
\date{} 
\begin{document}
\maketitle
\begin{abstract}
\citet{arnold2020bivariate} introduced a novel bivariate conditionally specified distribution, a distribution in which dependence between two random variables is established by defining the distribution of one variable conditional on the other. This novel conditioning regime was achieved through the use of survival functions, and the approach was termed the accelerated failure conditionals model. In their work, the conditioning framework was constructed using the exponential distribution. Although further generalization was proposed, challenges emerged in deriving the necessary and sufficient conditions for valid joint survival functions. The present study achieves such generalization, extending the conditioning framework to encompass distributional families whose marginal densities may exhibit unimodality and skewness, moving beyond distributional families whose marginal densities are non-increasing. The resulting models are fully specified through closed-form expressions for their moments, with simulations implemented using either a copula-based procedure or the Metropolis-Hastings algorithm. Empirical applications to two datasets, each featuring variables which are unimodal and skewed, demonstrate that the models with flexible, non-monotonic marginal densities yield a superior fit relative to those models with marginal densities restricted to monotonically decaying forms.
\end{abstract}

\section{Introduction}
As per \citet{arnold2020bivariate}, the general accelerated failure conditionals model is of the form: 
\begin{align}
    \bar{F}_X(x) = P(X>x) = \bar{F}_0(x),\quad x>0, \label{eq: general marginal}
\end{align}
for some survival function $\bar{F}_0(x) \in [0, 1] \ \forall x> 0$, and for each $x>0$:
\begin{equation}
    P(Y>y\mid X>x) = \bar{F}_1\left(\beta(x)y\right), \quad x, y >0,\label{eq: general cond}
\end{equation}
for some survival function $\bar{F}_1(\beta(x)y) \in [0, 1] \ \forall x, y >0$ and a suitable acceleration function $\beta(x)$. As per \citet{arnold2020bivariate}, in the analysis of dependent lifetimes of components in a system, it is more appropriate to consider the conditional density of $Y$ given that the first component, with lifetime $X$, remains operational at time $x$. Hence why, the conditioning event is taken to be $\{X > x\}$, rather than conditioning on the exact value $X = x$. The joint survival function is then:
\begin{align}
    P(X>x, Y>y) = \bar{F}_0(x) \bar{F}_1\left(\beta(x)y\right), \quad x, y>0. \label{eq: general joint}
\end{align}
Assuming differentiability and $X,Y$ are continuous, we obtain the marginal densities:
\[
f_X(x) = -\frac{d}{dx}P(X>x) = -\frac{d}{dx}\bar{F}_0(x) = f_0(x),
\]
and since $\lim_{x\to 0^+}\bar{F}_0(x)=1$, where $\beta(0):=\lim_{x\to 0^+}\beta(x)$:
\[
P(Y>y) = \lim_{x\to 0^+}P(Y>y,X>x) = \bar{F}_1\bigl(\beta(0)y\bigr).
\]
Now, the density of $Y$ is:
\[
f_Y(y) = -\frac{d}{dy}P(Y>y) 
= -\frac{d}{dy}\bar{F}_1(\beta(0)y)
= \beta(0)\,f_1(\beta(0)y),
\]
where $f_1(t)=-\bar{F}_1'(t)$ denotes the density associated with $\bar{F}_1$. Hence the marginal distribution of $Y$ belongs to the same family as $\bar{F}_1$, but with its argument scaled by $\beta(0)$. For \eqref{eq: general joint} to be a valid survival function, it must have a non-negative 
mixed partial derivative. Denoting $f_0$ and $f_1$ as the densities corresponding to the survival function $\bar{F}_0$ and $\bar{F}_1$, and differentiating, we obtain:
\begin{align}
    \frac{\partial}{\partial x} \frac{\partial}{\partial y} P(X>x, Y>y)  &= 
f_0(x) f_1(\beta(x) y)\beta(x)
- \bar{F}_0(x) f_1'(\beta(x) y) \beta(x) \beta'(x) y
- \bar{F}_0(x) f_1(\beta(x) y) \beta'(x) \nonumber \\
&= f_0(x) f_1(\beta(x) y)\beta(x)
- \beta'(x) \bar{F}_0(x) \left[  f_1(\beta(x) y)  +  f_1'(\beta(x) y) \beta(x)  y\right]. \nonumber
\end{align}
Denoting $t = \beta(x)y$, we have:
\begin{align}
    \frac{\partial}{\partial x} \frac{\partial}{\partial y} P(X>x, Y>y) 
&= f_0(x) f_1(t)\beta(x)
- \beta'(x) \bar{F}_0(x) \left[  f_1(t)  +  f_1'(t)t\right] \nonumber \\
& =  f_0(x) f_1(t)\beta(x)
- \beta'(x) \bar{F}_0(x) \left[ t f_1(t)  \right]' .\nonumber 
\end{align}
Thus a necessary and sufficient condition for $   \frac{\partial}{\partial x} \frac{\partial}{\partial y} P(X>x, Y>y)  \geq 0$ is:
\[
\beta'(x)\;\begin{cases}
\le\; \dfrac{f_0(x)}{\bar F_0(x)}\;\dfrac{f_1(t)}{[t f_1(t)]'}\;\beta(x), & \text{if } [t f_1(t)]'>0,\\[10pt]
\ge\; \dfrac{f_0(x)}{\bar F_0(x)}\;\dfrac{f_1(t)}{[t f_1(t)]'}\;\beta(x), & \text{if } [t f_1(t)]'<0, \\
\text{unconstrained}, & \text{if } [t f_1(t)]'=0.
\end{cases}
\]
Denoting the hazard function $h_0(x) = \frac{f_0(x)}{\bar{F}_0(x)}$ and $S(t) = \frac{f_1(t)}{[t f_1(t)]'}$, we may write the bounds as such:
\begin{align}
\sup_{t: [t f_1(t)]'<0} h_0(x) S(t)\beta(x) 
\;\leq\; \beta'(x) 
\;\leq\;
\inf_{t: [t f_1(t)]'>0} h_0(x) S(t)\beta(x). \label{eq: B'(x) bounds}
\end{align}
A special case of the upper bound for $B'(x)$ arises for distribution families satisfying $f_1'(t) \leq0$ - that is, for non-increasing densities. In this case, we note that $S(t) = \frac{f_1(t)}{\left[tf_1(t) \right]'} = \frac{f_1(t)}{f_1(t)  + tf_1'(t) } \geq1$ for $t \geq 0$, hence  $\beta'(x) 
\;\leq\;
 h_0(x)\beta(x)$.\\
 
It should be noted that multiple functional forms for $\beta(x)$ are possible. The question of selecting a form of $\beta(x)$ that provides the best model fit has already been investigated by \citet{arnold2020bivariate}. Accordingly, the present study restricts attention to those choices of $\beta(x)$ that arise from the extreme case in which $\beta'(x)$ attains its theoretical upper bound, that is, given $\inf_{t: [t f_1(t)]'>0} S(t) = c^*$:
\begin{align}
    \beta'(x) &= c^*h_0(x) \beta(x), \nonumber \\
    \therefore \int \frac{1}{\beta(x)} d\beta(x) &=  \int c^* h_0(x) dx, \nonumber\\
    \therefore \log\left( \beta(x) \right) &= -c^* \log\left( \bar{F}_0 (x)\right) + K. \nonumber \\
    \therefore \beta(x) &= \gamma \frac{1}{\left[\bar{F}_0(x)\right]^{c^*}}, \label{eq: beta(x) general}
    \end{align}
for $\gamma = e^K>0$. This suggests $\beta(x)$ is non-decreasing for $c^*>0$. Furthermore, $\beta(0): = \lim_{x\to 0^+} \beta(x) = \gamma$ since $\lim_{x\to 0^+}\bar{F}_0(x) = 1$. We also note that since $\bar{F}_0(x)$ approaches zero only as $x\rightarrow \infty$, the expression in \eqref{eq: beta(x) general} remains finite for all finite $x$.\\

The role of the acceleration function $\beta(x)$ in determining the sign of the correlation between $X$ and $Y$ is established in Theorem~\ref{thm: cov(x, y)}. In particular, since the present study considers only non-decreasing specifications of $\beta(x)$, the resulting framework is restricted to modelling data in which $X$ and $Y$ exhibit negative correlation.
\begin{theorem} \label{thm: cov(x, y)}
Consider an accelerated conditional model of the form: 
\[
P(X>x) = \bar{F}_0(x)
\quad \text{and} \quad
P(Y>y \mid X>x) = \bar{F}_1\bigl(\beta(x)y\bigr),
\]
where $\bar{F}_0$ and $\bar{F}_1$ are survival functions. Then:
\[
\beta'(x) \ge 0 \quad \Longrightarrow \quad {Cov}(X,Y) \le 0,
\qquad
\beta'(x) \le 0 \quad \Longrightarrow \quad {Cov}(X,Y) \ge 0.
\]
\end{theorem}

\begin{proof}
Since $P(Y>y , X> x) = \bar{F}_0(x)\bar{F}_1\bigl(\beta(x)y\bigr)$ then $P(Y>y) = \lim_{x\to 0^+}P(Y>y , X> x) = \bar{F}_0(0)\bar{F}_1\bigl(\beta(0)y\bigr) = \bar{F}_1\bigl(\beta(0)y\bigr)$ since $\lim_{x\to 0^+}\bar{F}_0(x) = 1$ and $\beta(0): = \lim_{x\to 0^+} \beta(x)$. We note from Hoeffding's covariance identity:
\begin{align}
    Cov(X, Y) & = \int_0^\infty \int_0^\infty\left( P(Y>y , X> x) - P(X> x)P(Y> y)   \right) \ dx dy \nonumber \\
    & = \int_0^\infty \int_0^\infty  \bar{F}_0(x) \left( \bar{F}_1\bigl(\beta(x)y\bigr) - \bar{F}_1\bigl(\beta(0)y\bigr)  \right) \ dx dy.\nonumber 
\end{align}
Now since $\bar{F}_1'(\cdot)\le 0$ (as $\bar{F}_1$ is a survival function), if $\beta'(x) \geq 0$ then $\bar{F}_1\bigl(\beta(x)y\bigr) \leq \bar{F}_1\bigl(\beta(0)y\bigr)$ and $Cov(X, Y) \leq 0$ for $x, y \geq 0$. Conversely, if $\beta'(x) \leq 0$ then $\bar{F}_1\bigl(\beta(x)y\bigr) \geq \bar{F}_1\bigl(\beta(0)y\bigr)$ and $Cov(X, Y) \geq 0$ for $x, y \geq 0$. Finally, if $\beta(x)$ is constant, say $\beta(x)\equiv c$, then $P(X>x,Y>y)=\bar F_0(x)\,\bar F_1(cy)=P(X>x)\,P(Y>y)$, so $X$ and $Y$ are independent and ${Cov}(X,Y)=0$.
\end{proof}
Furthermore, as done in \citet{arnold2020bivariate}, we introduce a dependence parameter $\tau$ which controls the strength of the dependence between $X$ and $Y$. Extending from \eqref{eq: beta(x) general}, we define the acceleration function as such:
\begin{align}
    \beta(x) &= \gamma \frac{1}{\left[\bar{F}_0(\tau x)\right]^{c^*}}. \label{eq: beta(x) general 2}
\end{align}
Given the bounds established in \eqref{eq: B'(x) bounds}, and upon selecting a non-decreasing specification for $\beta(x)$ as defined in \eqref{eq: beta(x) general 2}, and again assuming $\inf_{t: [t f_1(t)]'>0} S(t) = c^* > 0$, we obtain bounds for $\tau$ as such:
\begin{align}
    0 &\leq \beta'(x) \leq c^* h_0(x)\, \beta(x), \nonumber\\
    \therefore  
    0 &\leq 
    \frac{c^* \tau \gamma f_0(\tau x)}{\big[\bar{F}_0(\tau x)\big]^{c^* + 1}} 
    \leq 
    c^* \gamma \, 
    \frac{f_0(x)}{\big[\bar{F}_0(x)\big]^{c^*}}
    \frac{1}{\big[\bar{F}_0(\tau x)\big]^{c^*}}, \nonumber\\
    \therefore  
    0 &\leq \tau \leq \frac{h_0(x)}{h_0(\tau x)}. \nonumber\end{align}
We re-write this inequality as such: 
\begin{align}
    0 \le r(x, \tau) \le 1,
    \label{eq: tau bounds}
\end{align}
where $r(x, \tau) = \tau \frac{h_0(\tau x)}{h_0(x)}$, where the endpoint $\tau=0$ is included by continuity, with $r(x,0):=\lim_{\tau\to 0^+}r(x,\tau)=0$. Theorem \ref{thm: ABC} and Theorems \ref{thm: ifr hazard} and \ref{thm: factor hazard} in Appendix \ref{app: thms}, serve as a basis for establishing bounds on the dependence parameter $\tau$ that satisfy \eqref{eq: tau bounds}.

\begin{theorem}
\label{thm: ABC}
Let $h_0:(0,\infty)\to(0,\infty)$ be a hazard function. For $x>0$ and $\tau>0$ define:
\[
r(x,\tau):=\tau\,\frac{h_0(\tau x)}{h_0(x)},\qquad 
\phi(x):=x\,h_0(x).
\]
Let
\[
A:\ \ 0\le r(x,\tau)\le 1\ \ \text{for all }x>0,\qquad
B:\ \ \tau\in(0,1],\qquad
C:\ \ \phi \text{ is non-decreasing on }(0,\infty).
\]
Then:
\begin{enumerate}
\item[(1)] $(A \wedge B)\ \Rightarrow\ C$.
\item[(2)] $(C \wedge B)\ \Rightarrow\ A$.
\item[(3)] $(C \wedge A)\ \Rightarrow\ B$, \emph{unless} $\phi$ is identically constant on $(0,\infty)$, in which case $A$ holds for all $\tau>0$ and $B$ cannot be inferred.
\end{enumerate}

\end{theorem}

\begin{proof}
\textbf{(1) $(A \wedge B)\Rightarrow C$.}
Fix $0<y<x$ and set $\tau:=y/x\in(0,1]$ so that $B$ holds. By $A$,
\[
1 \ \ge\ r(x,\tau)=\tau\,\frac{h_0(\tau x)}{h_0(x)}=\frac{y\,h_0(y)}{x\,h_0(x)}.
\]
Hence $y\,h_0(y)\le x\,h_0(x)$ for all $0<y<x$, i.e.\ $\phi$ is non-decreasing. Thus $C$ holds.

\medskip
\textbf{(2) $(C \wedge B)\Rightarrow A$.}
Assume $C$ and $\tau\in(0,1]$. Then $\tau x\le x$ and, by the non-decreasing nature of $\phi$,
\[
\phi(\tau x)\le \phi(x)\ \Longleftrightarrow\ \tau x\,h_0(\tau x)\le x\,h_0(x)
\ \Longleftrightarrow\ r(x,\tau)\le 1.
\]
Since $h_0>0$, we also have $r(x,\tau)\ge 0$. Hence $A$ holds.

\medskip
\textbf{(3) $(C \wedge A)\Rightarrow B$, unless $\phi$ is constant.} 
Assume $C$ and $A$ hold. Because $h_0(x)>0$, the lower bound $r(x,\tau)\ge 0$ implies $\tau \ge 0$. Now suppose, toward a contradiction, that $\tau>1$, then $\tau x > x$ and, by assuming $\phi$ is now strictly increasing:
\[
\phi(\tau x) > \phi(x)\ \Longleftrightarrow\ \tau x\,h_0(\tau x) > x\,h_0(x)
\ \Longleftrightarrow\ r(x,\tau) > 1.
\]
contradicting the upper bound. Hence $\tau \le 1$. If instead $\phi$ is identically constant,
say $\phi(x)\equiv K>0$, then $h_0(x)=K/x$ and
\[
r(x,\tau)=\tau\,\frac{K/(\tau x)}{K/x}=1 \quad (\forall x>0,\ \forall \tau>0),
\]
so $A$ holds for all $\tau>0$ and no restriction $B$ follows.
\end{proof}
\begin{remark} When convenient, extend $B$ to $\tau\in[0,1]$ by defining $r(x,0):=\lim_{\tau\to 0^+} r(x,\tau)=0$ (whenever the limit exists). All implications above then carry over with this convention.
\end{remark}
\begin{remark}
If $h_0$ is non-decreasing, then $\phi(x) =  x h_0(x)$ is non-decreasing (product of increasing, positive functions), so Theorem~\ref{thm: ABC} recovers Theorem \ref{thm: ifr hazard}.
If $h_0(x)=C\,x^{p-1}g(x)$ with $C>0$, $p\ge 0$, and $g(x)$ non-decreasing, then $\phi(x) = C x^pg(x)$ is non-decreasing, so Theorem \ref{thm: ABC} recovers Theorem \ref{thm: factor hazard}.
\end{remark}

The following sections present the range of models employed in this study. In each case, the survival function forms are specified according to a single particular rate-shape distributional family such that $X \sim Family (\alpha, \lambda)$ and $Y \sim Family (\beta(0), \nu)$ where $\alpha$ and $\beta(0)$ denote the rate parameters and $\lambda$ and $\nu$ denote the shape parameters (shape parameters are not applicable for exponential and half-Cauchy models, however). Closed-form expressions of the models' moments are derived, and additionally, theoretical correlation bounds are obtained (with the exception of the half-Cauchy and gamma models).

\section{Exponential}
We restrict the survival function forms to the exponential family as such:
\begin{align}
    \bar{F}_0(x) &= P(X>x) = e^{-\alpha x}, \quad x> 0,\nonumber
\end{align}
where $\alpha > 0$ and, for each $x>0$:
\begin{align}
    \bar{F}_1(\beta(x) y) &= P(Y>y \mid X>x) = e^{-\beta(x) y}, \quad y>0, \nonumber
\end{align}
where $\beta(x) > 0$ for all $x>0 $ to ensure a valid survival function. The corresponding joint survival function will be:
\begin{align}
    \bar{F}_0(x) \bar{F}_1(\beta(x) y) &= P(X>x,Y>y) = e^{-\alpha x-\beta(x)y}, \quad x, y>0. \label{eq: exp joint survival}
\end{align}
Accordingly, $f_X(x) = \alpha e^{-(\alpha x)}$  for $x> 0$ and $f_Y(y) = \beta(0) e^{-(\beta(0)y )}$ for $y> 0$. Hence, $X\sim Exp(\alpha)$ and $Y\sim Exp(\beta(0))$.\\

Now since $[tf_1(t)]' = (1-t)e^{-t} > 0$ for $t<1$, and $S(t) = -\frac{1}{t-1}$ is increasing on this interval, $\inf_{t<1} S(t) = \lim_{t \to 0^+} S(t) = c^*= 1$. Additionally, since $[tf_1(t)]' = (1-t)e^{-t} < 0$ for $t>1$, and $S(t) = -\frac{1}{t-1}$ is increasing on this interval, $\sup_{t>1} S(t) = \lim_{t \to \infty} S(t) = 0$. Consequently, applying \eqref{eq: B'(x) bounds} yields:
\begin{align}
    0 \leq B'(x) \leq h_0(x)\beta(x) = \alpha \beta(x), \nonumber
\end{align}
corresponding to Proposition $3.1$ in \citet{arnold2020bivariate}. Now, we select the form of the acceleration function $\beta(x)$ such that $\beta'(x) = \alpha\beta(x)$ (that is, equal to $\beta'(x)'s$ upper bound). Hence from \eqref{eq: beta(x) general} with $c^* = 1$, $\beta(x) = \gamma e^{\alpha x}$ with $\gamma >0$. Furthermore, following \citet{arnold2020bivariate}, we introduce a dependence parameter $\tau$ which governs the strength of the dependence between $X$ and $Y$. From Theorem \ref{thm: ifr hazard}, since $h_0(x) = \alpha$ is non-decreasing, in order for condition \eqref{eq: tau bounds} to hold, we must have $\tau \in [0, 1]$. Accordingly, we select our acceleration function as:
\begin{align}
    \beta(x) =  \gamma e^{\alpha\tau x},\quad x > 0, \nonumber
\end{align}
where $\alpha, \gamma > 0$ with $\tau \in [0, 1]$.
\subsection{Moments}
Noting $X \sim Exp(\alpha)$ and $Y \sim Exp(\gamma)$, we have:
\begin{align}
    E(X) &= \frac{1}{\alpha}, \label{eq: exp ex}  \\
    E(Y) &= \frac{1}{\gamma}, \label{eq: exp ey}  \\
    Var(X) & = \frac{1}{\alpha^2}, \label{eq: exp varx} \\
    Var(Y) &= \frac{1}{\gamma^2}. \label{eq: exp vary}
\end{align}
Accordingly:
\begin{align}
        Cov(X, Y) & = \int_0^\infty \int_0^\infty\left( P(Y>y , X> x) - P(X> x)P(Y> y)   \right) \ dx dy \nonumber \\
    & = \int_0^\infty \int_0^\infty  \bar{F}_0(x) \left( \bar{F}_1\bigl(\beta(x)y\bigr) - \bar{F}_1\bigl(\beta(0)y\bigr)  \right) \ dx dy \nonumber \\
    & = \int_0^\infty \int_0^\infty e^{-(\alpha x)} \left(e^{-\beta(x)y} - e^{-\beta(0)y} \right) \ dy dx \nonumber\\
    & = \int_0^\infty e^{-(\alpha x)} \left(\frac{1}{\beta(x)} - \frac{1}{\beta(0)} \right) \ dx \nonumber \\
    & =  \frac{1}{\gamma} \int_0^\infty \left( e^{-\alpha(1+\tau)x} - e^{-(\alpha x)}\right) \ dx \nonumber \\
    & = -\frac{\tau}{\alpha \gamma(1+\tau)}. \label{eq: exp cov}
\end{align}
Using \eqref{eq: exp varx}, \eqref{eq: exp vary} and \eqref{eq: exp cov}, we have the correlation between $X$ and $Y$ as:
\begin{align}
    \rho(X, Y) &= \frac{Cov(X, Y)}{\sqrt{Var(X)Var(Y)}} \nonumber \\
    &= -\frac{\tau}{1+\tau}. \nonumber
\end{align}
Seeing as $\tau\in [0, 1]$, this exponential model with particular acceleration acceleration function $\beta(x) = \gamma e^{\alpha\tau x}$, may only be used to model data for which $\rho(X, Y) \in [-\frac{1}{2}, 0]$.\\

Additionally, to facilitate simulation of a $Y$ conditional on $X$, it is convenient to formally note the conditional cumulative distribution function $F_{Y \mid X=x}$ as:
\begin{align}
    F_{Y \mid X = x}(y) 
    &= \int_{0}^{y} \frac{f_{X,Y}(x, t)}{f_X(x)} \, dt \nonumber \\
    &=(-\gamma\tau y e^{\alpha \tau x} + e^{\gamma y e^{\alpha \tau x}} - 1)e^{-\gamma y e^{\alpha \tau x}}, \label{eq: Fygivx exp}
\end{align}
where $f_{X,Y}(x,y)$  is provided in Appendix~\ref{app: joint lik}.

\section{Lomax}
Assuming the survival function forms belong to the Lomax family, we have:
\begin{align}
    \bar{F}_0(x) &= P(X>x) = (1+\alpha x)^{-\lambda}, \quad x> 0,\nonumber
\end{align}
where $\alpha,\lambda > 0$ and, for each $x>0$:
\begin{align}
    \bar{F}_1(\beta(x) y) &= P(Y>y \mid X>x) = (1+\beta(x) y)^{-\nu}, \quad y>0, \nonumber
\end{align}
where $\beta(x), \nu > 0$ for all $x>0 $ to ensure a valid survival function. The corresponding joint survival function will be:
\begin{align}
    \bar{F}_0(x) \bar{F}_1(\beta(x) y) &= P(X>x,Y>y) = (1+\alpha x)^{-\lambda}(1+\beta(x) y)^{-\nu}, \quad x, y>0. \label{eq: Lomax joint survival}
\end{align}
Accordingly, $f_X(x) = \alpha \lambda (\alpha x+ 1)^{-(\alpha+1)}$ for $x> 0$ and $f_Y(y) =  \beta(0) \nu (\beta(0)y + 1)^{-(\nu+1)}$ for $y> 0$. Hence, $X\sim Lomax(\alpha, \lambda)$ and $Y\sim Lomax(\beta(0), \nu)$.\\

Now since $[tf_1(t)]' = \frac{\nu(t+1)^{-\nu}(1-\nu t)}{t^2+2t+1} > 0$ for $t<\frac{1}{v}$, and $S(t) = -\frac{t+1}{vt-1}$ is increasing on this interval, $\inf_{t<\frac{1}{v}} S(t) = \lim_{t \to 0^+} S(t) =c^*= 1$. Additionally, since $[tf_1(t)]' = \frac{\nu(t+1)^{-\nu}(1-\nu t)}{t^2+2t+1}< 0$ for $t>\frac{1}{\nu}$, and $S(t) =  -\frac{t+1}{vt-1}$ is increasing on this interval, $\sup_{t>\frac{1}{\nu}} S(t) = \lim_{t \to \infty} S(t) = -\frac{1}{\nu}$. Consequently, applying \eqref{eq: B'(x) bounds} yields:
\begin{align}
    - \frac{1}{\nu} \frac{\alpha \lambda}{\alpha x+1} \beta(x) \leq B'(x) \leq \frac{\alpha \lambda}{\alpha x+1} \beta(x). \nonumber
\end{align}
Now, we select the form of the acceleration function $\beta(x)$ such that $\beta'(x) =  \frac{\alpha \lambda}{\alpha x+1} \beta(x)$ (that is, equal to $\beta'(x)'s$ upper bound). Hence from \eqref{eq: beta(x) general} with $c^* = 1$, $\beta(x) = \gamma \left(1+\alpha x \right)^{\lambda}$ with $\gamma >0$. From Theorem~\ref{thm: factor hazard}, since we can express the hazard function as $h_0(x) = \frac{\alpha\lambda}{\alpha x + 1} 
= C\,x^{p-1} g(x)$ with $C = \alpha\lambda, p = 1$ and $g(x) = \frac{1}{\alpha x + 1}$ where $g(x)$ is a non-decreasing function for $x>0$, it follows that in order for condition \eqref{eq: tau bounds} to hold, we must have $\tau \in [0, 1]$. Accordingly, we select our acceleration function as:
\begin{align}
    \beta(x) =  \gamma \left(1+\alpha\tau x \right)^{\lambda} ,\quad x > 0, \nonumber
\end{align}
where $\alpha, \lambda, \gamma, \nu  > 0$ with $\tau \in [0, 1]$.

\subsection{Moments}
Noting $X \sim Lomax(\alpha, \lambda)$ and $Y \sim Lomax(\gamma, \nu)$, we have:
\begin{align}
    E(X) &= \frac{1}{\alpha(\lambda - 1)}, \quad \lambda > 1, \label{eq: Lomax ex} \\
    E(Y) &= \frac{1}{\gamma(\nu - 1)}, \quad \nu > 1, \label{eq: Lomax ey}  \\
    {Var}(X) &= \frac{\lambda}{\alpha^2(\lambda-1)^2(\lambda - 2)}, \quad \lambda > 2, \label{eq: Lomax varx} \\
    {Var}(Y) &= \frac{\nu}{\gamma^2(\nu - 1)^2(\nu -2)}, \quad \nu > 2. \label{eq: Lomax vary}
\end{align}
Accordingly:
\begin{align}
\operatorname{Cov}(X, Y) 
&= \int_0^\infty \int_0^\infty 
   \left( P(Y>y , X> x) - P(X> x)P(Y> y) \right) \,dy\,dx \nonumber \\
&= \int_0^\infty \int_0^\infty  
   \bar{F}_0(x) \left( \bar{F}_1\bigl(\beta(x)y\bigr) - \bar{F}_1\bigl(\beta(0)y\bigr) \right) \,dy\,dx \nonumber \\
&= \int_0^\infty (1+\alpha x)^{-\lambda} 
   \left( \int_0^\infty \left((1+\beta(x) y)^{-\nu} - (1+\beta(0) y)^{-\nu} \right) dy \right) dx \nonumber \\
&=  \frac{1}{\nu - 1 }\int_0^\infty (1+\alpha x)^{-\lambda} 
   \left( \frac{1}{\beta(x)} - \frac{1}{\beta(0)} \right) dx \nonumber \\
&=  \frac{1}{\gamma(\nu -1)} 
   \left[\int_0^\infty \frac{dx}{(1+\alpha x)^\lambda(1+\alpha \tau x)^\lambda}   
         - \int_0^\infty  \frac{dx}{(1+\alpha x)^\lambda}\right]  \nonumber \\
&= \frac{1}{\alpha \gamma (\nu -1)}
   \left[ \frac{1}{2\lambda - 1} \, {}_{2}F_1\!\left( \lambda, 1; 2\lambda; 1 -\tau \right) 
         - \frac{1}{\lambda - 1}\right], 
\label{eq: Lomax cov}
\end{align}
where the conditions $\nu>1$ and $\lambda>1$ ensure convergence of the inner and outer integrals, respectively. Now, using \eqref{eq: Lomax varx}, \eqref{eq: Lomax vary} and \eqref{eq: Lomax cov}, the correlation function is:
\begin{align}
    \rho(X,Y)
    &= \frac{\operatorname{Cov}(X,Y)}{\sqrt{\operatorname{Var}(X)\operatorname{Var}(Y)}} \nonumber\\
    &= (\lambda-1)\left(\frac{{}_{2}F_{1}\!\left(\lambda,1;2\lambda;1-\tau\right)}{2\lambda-1} - \frac{1}{\lambda-1}\right)
       \sqrt{\frac{(\lambda-2)(\nu-2)}{\lambda\,\nu}},
    \label{eq: Lomax rho}
\end{align}
valid for $\lambda,\nu>2$. Since ${}_{2}F_{1}(\lambda,1;2\lambda;z)$ has a power series with strictly positive coefficients, it is strictly increasing in $z\in[0,1]$; hence ${}_{2}F_{1}\!\left(\lambda,1;2\lambda;1-\tau\right)$ is strictly decreasing in $\tau\in[0,1]$. Consequently, \eqref{eq: Lomax rho} attains its minimum at $\tau=1$ and its maximum (zero) at $\tau=0$. Therefore,
\begin{align}
    -\,\frac{\lambda}{2\lambda-1}\,\sqrt{\frac{\lambda-2}{\lambda}}\,
      \sqrt{\frac{\nu-2}{\nu}}\le \rho(X,Y) \le 0.\nonumber
\end{align}
Moreover, $\rho_{\min}(\lambda,\nu):=-\tfrac{\lambda}{2\lambda-1}
\sqrt{\tfrac{\lambda-2}{\lambda}}\sqrt{\tfrac{\nu-2}{\nu}}$ is strictly
decreasing in each of $\lambda$ and $\nu$, so
$\inf_{\lambda>2,\;\nu>2}\rho_{\min}(\lambda,\nu)\allowbreak
=\allowbreak \lim_{\lambda,\nu\to\infty}\rho_{\min}(\lambda,\nu)\allowbreak
=\allowbreak -\tfrac{1}{2}$. Thus, as in the exponential case, this Lomax model with acceleration $\beta(x)=\gamma\,(1+\alpha\tau x)^{\lambda}$ can only accommodate correlations in the range $\rho(X,Y)\in\big[-\tfrac{1}{2},\,0\big]$.

\section{Weibull}
Restricting the survival function forms to the Weibull family, we obtain:
\begin{align}
    \bar{F}_0(x) &= P(X>x) = e^{-(\alpha x)^\lambda}, \quad x> 0,\nonumber
\end{align}
where $\alpha,\lambda > 0$ and, for each $x>0$:
\begin{align}
    \bar{F}_1(\beta(x) y) &= P(Y>y \mid X>x) = e^{-(\beta(x) y)^\nu}, \quad y>0, \nonumber
\end{align}
where $\beta(x), \nu > 0$ for all $x>0 $ to ensure a valid survival function. The corresponding joint survival function will be:
\begin{align}
    \bar{F}_0(x) \bar{F}_1(\beta(x) y) &= P(X>x,Y>y) = e^{-(\alpha x)^\lambda - (y\beta(x))^\nu}, \quad x, y>0. \label{eq: Weibull joint survival}
\end{align}
Accordingly, $f_X(x) = \alpha \lambda (\alpha x)^{\lambda - 1}e^{-(\alpha x )^\lambda}$ for $x> 0$ and $f_Y(y) =  \beta(0) \nu (\beta(0) y)^{\nu - 1}e^{-(\beta(0) y )^\nu}$ for $y> 0$. Hence, $X\sim Weibull(\alpha, \lambda)$ and $Y\sim Weibull(\beta(0), \nu)$.\\

Now since $[tf_1(t)]' = \nu^2t^{\nu - 1} (1- t^\nu)e^{-t^\nu} > 0$ for $t<1$, and $S(t) = \frac{1}{\nu (1- t^\nu)}$ is increasing on this interval, $\inf_{t<1} S(t) = \lim_{t \to 0^+} S(t) =c^*= \frac{1}{\nu}$. Additionally, since $[tf_1(t)]' = \nu^2t^{\nu - 1} (1- t^\nu)e^{-t^\nu}  < 0$ for $t>1$, and $S(t) = \frac{1}{\nu (1- t^\nu)}$ is increasing on this interval, $\sup_{t>1} S(t) = \lim_{t \to \infty} S(t) = 0$. Consequently, applying \eqref{eq: B'(x) bounds} yields:
\begin{align}
    0 \leq B'(x) \leq \frac{1}{\nu}\lambda\alpha (\alpha x)^{\lambda - 1} \beta(x). \nonumber
\end{align}
Now, we select the form of the acceleration function $\beta(x)$ such that $\beta'(x) =  \frac{1}{\nu}\lambda\alpha (\alpha x)^{\lambda - 1} \beta(x)$ (that is, equal to $\beta'(x)'s$ upper bound). Hence from \eqref{eq: beta(x) general} with $c^* = \frac{1}{\nu}$, $\beta(x) = \gamma e^{\frac{1}{\nu}(\alpha x)^\lambda}$ with $\gamma >0$. From Theorem~\ref{thm: factor hazard}, since we can express the hazard function as $h_0(x) = \lambda\alpha(\alpha x)^{\lambda - 1}
= C\,x^{p-1} g(x)$ with $C = \lambda\alpha^\lambda, p = \lambda$ and $g(x) = 1$
where $g(x)$ is a non-decreasing function for $x>0$, it follows that in order for condition \eqref{eq: tau bounds} to hold, we must have $\tau \in [0, 1]$. Accordingly, we select our acceleration function as:
\begin{align}
    \beta(x) =  \gamma e^{\frac{1}{\nu}(\alpha \tau x)^\lambda},\quad x > 0, \nonumber
\end{align}
where $\alpha, \lambda, \gamma, \nu  > 0$ with $\tau \in [0, 1]$.

\subsection{Moments}
Noting $X \sim Weibull(\alpha, \lambda)$ and $Y \sim Weibull(\gamma, \nu)$, we have:
\begin{align}
    E(X) &= \frac{1}{\alpha} \Gamma\left(1+ \frac{1}{\lambda}\right), \label{eq: Weibull ex} \\
    E(Y) &= \frac{1}{\gamma} \Gamma\left(1+ \frac{1}{\nu}\right), \label{eq: Weibull ey}  \\
    {Var}(X) &=  \frac{1}{\alpha^2}\left( \Gamma\left( 1+ \frac{2}{\lambda}\right)  - \left(\Gamma\left(1+\frac{1}{\lambda}\right)\right)^2 \right),\label{eq: Weibull varx} \\
    {Var}(Y) &=  \frac{1}{\gamma^2}\left( \Gamma\left( 1+ \frac{2}{\nu}\right)  - \left(\Gamma\left(1+\frac{1}{\nu}\right)\right)^2 \right).\label{eq: Weibull vary} 
\end{align}
Accordingly:
\begin{align}
\operatorname{Cov}(X, Y) 
&= \int_0^\infty \int_0^\infty 
   \big( P(Y>y , X> x) - P(X> x)P(Y> y) \big) \,dy\,dx \nonumber \\
&= \int_0^\infty \int_0^\infty  
   \bar{F}_0(x) \left( \bar{F}_1\bigl(\beta(x)y\bigr) - \bar{F}_1\bigl(\beta(0)y\bigr) \right) \,dy\,dx \nonumber \\
&= \int_0^\infty e^{-(\alpha x)^\lambda}
   \left( \int_0^\infty \left( e^{-(\beta(x)y )^\nu}- e^{-(\beta(0)y )^\nu} \right) dy \right) dx \nonumber \\
&= \int_0^\infty e^{-(\alpha x)^\lambda}
   \left( \frac{\Gamma\!\left(1+\frac{1}{\nu}\right)}{\beta(x)}  -\frac{\Gamma\!\left(1+\frac{1}{\nu}\right)}{\beta(0)}  \right) dx \nonumber \\
&=  \frac{1}{\gamma } \Gamma\!\left( 1+ \frac{1}{\nu}\right) 
   \left[\int_0^\infty e^{-(\alpha x)^\lambda - \frac{1}{\nu}(\alpha \tau x)^\lambda} \, dx 
         - \int_0^\infty e^{-(\alpha x)^\lambda} \, dx \right] \nonumber \\
&= \frac{1}{\alpha \gamma} \Gamma\!\left(1+ \frac{1}{\nu} \right)\Gamma\!\left(1+ \frac{1}{\lambda} \right)
   \left[ \left(\frac{\nu}{\nu + \tau^\lambda} \right)^{\frac{1}{\lambda}} - 1 \right].
\label{eq: Weibull cov}
\end{align}
Now, using \eqref{eq: Weibull varx}, \eqref{eq: Weibull vary} and \eqref{eq: Weibull cov}, the correlation function is:
\begin{align}
    \rho(X,Y)
    &= \frac{\operatorname{Cov}(X,Y)}{\sqrt{\operatorname{Var}(X)\operatorname{Var}(Y)}} \nonumber\\
    &= \frac{\Gamma\left(1+ \frac{1}{\nu} \right)\Gamma\left(1+ \frac{1}{\lambda} \right)\left[ \left(\frac{\nu}{\nu + \tau^\lambda} \right)^{\frac{1}{\lambda}} - 1 \right]}{\sqrt{\left( \Gamma\left( 1+ \frac{2}{\lambda}\right)  - \left(\Gamma\left(1+\frac{1}{\lambda}\right)\right)^2 \right)\left( \Gamma\left( 1+ \frac{2}{\nu}\right)  - \left(\Gamma\left(1+\frac{1}{\nu}\right)\right)^2 \right)}}.
    \label{eq: Weibull rho}
\end{align}
Clearly, \eqref{eq: Weibull rho} is strictly decreasing in $\tau \in [0, 1]$. Hence:
\begin{align}
     \rho_{\text{min}}(\lambda, \nu):=\frac{\Gamma\left(1+ \frac{1}{\nu} \right)\Gamma\left(1+ \frac{1}{\lambda} \right)\left[ \left(\frac{\nu}{\nu + 1} \right)^{\frac{1}{\lambda}} - 1 \right]}{\sqrt{\left( \Gamma\left( 1+ \frac{2}{\lambda}\right)  - \left(\Gamma\left(1+\frac{1}{\lambda}\right)\right)^2 \right)\left( \Gamma\left( 1+ \frac{2}{\nu}\right)  - \left(\Gamma\left(1+\frac{1}{\nu}\right)\right)^2 \right)}} \leq \rho(X, Y) \leq 0.
\end{align}
Now, $\inf_{\lambda, \nu > 0} \rho_{\text{min}}(\lambda, \nu) = \lim_{\nu \to \infty}\rho_{\text{min}}(\lambda = 1, \nu)  = -\frac{\sqrt{6}}{\pi}$. Hence, this Weibull model with acceleration function $\beta(x) = \gamma e^{\frac{1}{\nu}(\alpha \tau x)^\lambda}$ will be able to accommodate correlations in the range $\rho(X,Y)\in\big[-\tfrac{\sqrt{6}}{\pi},\,0\big]$.

\section{Log-Logistic}
We restrict the survival function forms to the log-logistic family as such:
\begin{align}
    \bar{F}_0(x) &= P(X>x) = 1 - \frac{1}{1+(\alpha x)^{-\lambda}} = \frac{1}{1+ (\alpha x)^\lambda}, \quad x> 0,\nonumber
\end{align}
where $\alpha,\lambda > 0$ and, for each $x>0$:
\begin{align}
    \bar{F}_1(\beta(x) y) &= P(Y>y \mid X>x) =1 - \frac{1}{1+(\beta(x) y)^{-\nu}} = \frac{1}{1+ (\beta(x) y)^\nu}, \quad y>0, \nonumber
\end{align}
where $\beta(x), \nu > 0$ for all $x>0 $ to ensure a valid survival function. The corresponding joint survival function will be:
\begin{align}
    \bar{F}_0(x) \bar{F}_1(\beta(x) y) &= P(X>x,Y>y) = \frac{1}{(1+(\alpha x)^\lambda)(1+(\beta(x) y)^\nu)}, \quad x, y>0. \label{eq: Log joint survival}
\end{align}
Accordingly, $f_X(x) = \frac{\alpha \lambda(\alpha x)^{\lambda-1}}{((\alpha x)^\lambda+1)^2}$ for $x> 0$ and $f_Y(y) = \frac{\beta(0) \nu(\beta(0) y)^{\nu-1}}{((\beta(0) y)^\nu+1)^2}$ for $y> 0$. Hence, $X\sim Log-Logistic(\alpha, \lambda)$ and $Y\sim Log-Logistic(\beta(0), \nu)$.\\

Now since $[tf_1(t)]' = \frac{\nu^2 t^{\nu -1} (1- t^\nu)}{t^{3\nu} + 3t^{2\nu} + 3t^{\nu}+1} > 0$ for $t<1$, and $S(t) = -\frac{t^\nu+1}{\nu (t^\nu -1)}$ is increasing on this interval, $\inf_{t<1} S(t) = \lim_{t \to 0^+} S(t) = c^*= \frac{1}{\nu}$. Additionally, since $[tf_1(t)]' =\frac{\nu^2 t^{\nu -1} (1- t^\nu)}{t^{3\nu} + 3t^{2\nu} + 3t^{\nu}+1} < 0$ for $t>1$, and $S(t) = -\frac{t^\nu+1}{\nu (t^\nu -1)}$ is increasing on this interval, $\sup_{t>1} S(t) = \lim_{t \to \infty} S(t) = -\frac{1}{\nu}$. Consequently, applying \eqref{eq: B'(x) bounds} yields:
\begin{align}
    -\frac{1}{\nu} \frac{\alpha \lambda(\alpha x)^{\lambda - 1}}{(\alpha x)^\lambda+1} \beta(x) \leq B'(x) \leq\frac{1}{\nu} \frac{\alpha \lambda(\alpha x)^{\lambda - 1}}{(\alpha x)^\lambda+1} \beta(x) . \nonumber
\end{align}
Now, we select the form of the acceleration function $\beta(x)$ such that $\beta'(x) = \frac{1}{\nu} \frac{\alpha \lambda(\alpha x)^{\lambda - 1}}{(\alpha x)^\lambda+1} \beta(x)$ (that is, equal to $\beta'(x)'s$ upper bound). Hence from \eqref{eq: beta(x) general} with $c^* = \frac{1}{\nu}$, $\beta(x) = \gamma (1+(\alpha x)^\lambda)^{\frac{1}{\nu}}$ with $\gamma >0$. From Theorem~\ref{thm: factor hazard}, since we can express the hazard function as $h_0(x) = \frac{\alpha \lambda(\alpha x)^{\lambda - 1}}{(\alpha x)^\lambda+1} 
= C\,x^{p-1} g(x)$ with $C = \lambda\alpha^\lambda, p = \lambda$ and $g(x) = \frac{1}{(\alpha x)^\lambda+1}$
where $g(x)$ is a non-decreasing function for $x>0$, it follows in order for condition \eqref{eq: tau bounds} to hold, we must have $\tau \in [0, 1]$. Accordingly, we select our acceleration function as:
\begin{align}
    \beta(x) =  \gamma (1+(\alpha \tau x)^\lambda)^{\frac{1}{\nu}},\quad x > 0, \nonumber
\end{align}
where $\alpha, \lambda, \gamma, \nu  > 0$ with $\tau \in [0, 1]$.

\subsection{Moments}
Noting $X \sim Log-Logistic(\alpha, \lambda)$ and $Y \sim Log-Logistic(\gamma, \nu)$, we have:
\begin{align}
    E(X) &= \frac{\pi}{\alpha \lambda \sin\!\left(\tfrac{\pi}{\lambda}\right)}, \quad \lambda > 1,  \label{eq: Log ex} \\
    E(Y) &= \frac{\pi}{\gamma \nu \sin\!\left(\tfrac{\pi}{\nu}\right)}, \quad \nu > 1,  \label{eq: Log ey}  \\
    {Var}(X) &= \frac{1}{\alpha^2} \left[ \frac{2\pi}{\lambda \sin\!\left(\tfrac{2\pi}{\lambda}\right)} 
            - \left(\frac{\pi}{\lambda \sin\!\left(\tfrac{\pi}{\lambda}\right)}\right)^2 \right], \quad \lambda > 2, 
            \label{eq: Log varx} \\
    {Var}(Y) &= \frac{1}{\gamma^2} \left[ \frac{2\pi}{\nu \sin\!\left(\tfrac{2\pi}{\nu}\right)} 
            - \left(\frac{\pi}{\nu \sin\!\left(\tfrac{\pi}{\nu}\right)}\right)^2 \right], \quad \nu > 2.
            \label{eq: Log vary}
\end{align}

Accordingly:
\begin{align}
\operatorname{Cov}(X, Y) 
&= \int_0^\infty \int_0^\infty 
   \big( P(Y>y , X> x) - P(X> x)P(Y> y) \big) \,dy\,dx \nonumber \\
&= \int_0^\infty \int_0^\infty  
   \bar{F}_0(x) \left( \bar{F}_1\bigl(\beta(x)y\bigr) - \bar{F}_1\bigl(\beta(0)y\bigr) \right) \,dy\,dx \nonumber \\
&= \int_0^\infty \frac{1}{1+(\alpha x)^\lambda}
   \left( \int_0^\infty \left( \frac{1}{1+(\beta(x) y)^\nu}- \frac{1}{1+(\beta(0) y)^\nu} \right) dy \right) dx \nonumber \\
&= \frac{\pi}{\nu \sin\left(\frac{\pi}{\nu} \right)}\int_0^\infty \frac{1}{1+(\alpha x)^\lambda}
   \left( \frac{1}{\beta(x)}  -\frac{1}{\beta(0)}  \right) dx \nonumber \\
&=  \frac{\pi}{\gamma \nu \sin\left(\frac{\pi}{\nu} \right)}
   \left[\int_0^\infty \frac{dx}{\left((1+(\alpha x )^\lambda \right)\left((1+(\alpha\tau x )^\lambda \right)^\frac{1}{\nu}}
         - \int_0^\infty  \frac{dx}{1+(\alpha x)^\lambda}\right] \nonumber \\
&= \frac{\pi}{\alpha \lambda \nu \gamma  \sin\left(\frac{\pi}{\nu}\right)}\left[ \frac{\Gamma\left( \frac{1}{\lambda}\right)\Gamma\left(1 +\frac{1}{\nu} - \frac{1}{\lambda} \right)}{\Gamma\left(1+\frac{1}{\nu} \right)} {}_{2}F_{1}\left( \frac{1}{\nu}, \frac{1}{\lambda}; 1+\frac{1}{\nu}; 1 - \tau^\lambda\right) - \frac{\pi}{\sin\left(\frac{\pi}{\lambda}\right)}  \right].
\label{eq: Log cov}
\end{align}
where the conditions $\nu>1$ and $\lambda>1$ ensure convergence of the inner and outer integrals, respectively. Now, using \eqref{eq: Log varx}, \eqref{eq: Log vary} and \eqref{eq: Log cov}, the correlation function is:
\begin{align}
    \rho(X,Y)
    &= \frac{\operatorname{Cov}(X,Y)}{\sqrt{\operatorname{Var}(X)\operatorname{Var}(Y)}} \nonumber\\
    &= \frac{\pi\left[ \frac{\Gamma\left( \frac{1}{\lambda}\right)\Gamma\left(1 +\frac{1}{\nu} - \frac{1}{\lambda} \right)}{\Gamma\left(1+\frac{1}{\nu} \right)} {}_{2}F_{1}\left( \frac{1}{\nu}, \frac{1}{\lambda}; 1+\frac{1}{\nu}; 1 - \tau^\lambda\right) - \frac{\pi}{\sin\left(\frac{\pi}{\lambda}\right)}  \right]}{\lambda \nu \sin\left( \frac{\pi}{\nu}\right) \sqrt{ \left[ \frac{2\pi}{\lambda \sin\!\left(\tfrac{2\pi}{\lambda}\right)} 
            - \left(\frac{\pi}{\lambda \sin\!\left(\tfrac{\pi}{\lambda}\right)}\right)^2 \right] \left[ \frac{2\pi}{\nu \sin\!\left(\tfrac{2\pi}{\nu}\right)} 
            - \left(\frac{\pi}{\nu \sin\!\left(\tfrac{\pi}{\nu}\right)}\right)^2 \right]}},
    \label{eq: Log rho}
\end{align}
valid for $\lambda, \nu >2$. Since $ {}_{2}F_{1}\left( \frac{1}{\nu}, \frac{1}{\lambda}; 1+\frac{1}{\nu}; z\right)$ has a power series with strictly positive coefficients, it is strictly increasing in $z\in[0,1]$; hence $ {}_{2}F_{1}\left( \frac{1}{\nu}, \frac{1}{\lambda}; 1+\frac{1}{\nu}; 1 - \tau^\lambda\right)$ is strictly decreasing in $\tau\in[0,1]$. Consequently, \eqref{eq: Log rho} attains its minimum at $\tau=1$ and its maximum (zero) at $\tau=0$. Therefore,
\begin{align}
    \rho_{\text{min}(\lambda, \nu)}:=\frac{\pi\left[ \frac{\Gamma\left( \frac{1}{\lambda}\right)\Gamma\left(1 +\frac{1}{\nu} - \frac{1}{\lambda} \right)}{\Gamma\left(1+\frac{1}{\nu} \right)}  - \frac{\pi}{\sin\left(\frac{\pi}{\lambda}\right)}  \right]}{\lambda \nu \sin\left( \frac{\pi}{\nu}\right) \sqrt{ \left[ \frac{2\pi}{\lambda \sin\!\left(\tfrac{2\pi}{\lambda}\right)} 
            - \left(\frac{\pi}{\lambda \sin\!\left(\tfrac{\pi}{\lambda}\right)}\right)^2 \right] \left[ \frac{2\pi}{\nu \sin\!\left(\tfrac{2\pi}{\nu}\right)} 
            - \left(\frac{\pi}{\nu \sin\!\left(\tfrac{\pi}{\nu}\right)}\right)^2 \right]}} \le \rho(X,Y) \le 0.\nonumber
\end{align}
Now denoting $[\lambda^*, \nu^*]' = \operatorname{argmin}_{\lambda, \nu} \rho_{\text{min}(\lambda, \nu)}$, we obtain $\inf_{\lambda, \nu >2}\rho_{\text{min}}(\lambda, \nu) = \lim_{\nu  \to \infty}\rho_{\text{min}}(\lambda = \lambda^* \approx 5.06, \nu) \approx -0.54076$. Hence, this log-logistic model with acceleration function $\beta(x) = \gamma (1+(\alpha \tau x)^\lambda)^{\frac{1}{\nu}}$ will be able to accommodate correlations in the range $\rho(X,Y)\in\big[-0.54076,\,0\big]$.

\section{Half-Cauchy}
Restricting the survival function forms to the half-Cauchy family, we obtain:
\begin{align}
    \bar{F}_0(x) &= P(X>x) = 1 - \frac{2}{\pi}\arctan(\alpha x), \quad x> 0,\nonumber
\end{align}
where $\alpha > 0$ and, for each $x>0$:
\begin{align}
    \bar{F}_1(\beta(x) y) &= P(Y>y \mid X>x) =1 - \frac{2}{\pi}\arctan(\beta(x) y), \quad y>0, \nonumber
\end{align}
where $\beta(x) > 0$ for all $x>0 $ to ensure a valid survival function. The corresponding joint survival function will be:
\begin{align}
    \bar{F}_0(x) \bar{F}_1(\beta(x) y) &= P(X>x,Y>y) = \frac{1}{\pi^2}\left(2\arctan(\alpha x) - \pi \right)\left(2\arctan(\beta(x) y) - \pi \right), \quad x, y>0. \label{eq: Cauchy joint survival}
\end{align}
Accordingly, $f_X(x) = \frac{2\alpha}{\pi (1+(\alpha x)^2)}$ for $x> 0$ and $f_Y(y) = \frac{2\beta(0)}{\pi (1+(\beta(0) y)^2)}$ for $y> 0$. Hence, $X\sim Half-Cauchy(\alpha)$ and $Y\sim Half-Cauchy(\beta(0))$.\\

Now since $[tf_1(t)]' = \frac{2(1 - t^2)}{\pi (t^2+1)^2} > 0 $ for $t<1$, and $S(t) = \frac{t^2+1}{1 - t^2}$ is increasing on this interval, $\inf_{t<1} S(t) = \lim_{t \to 0^+} S(t) = c^*=  1$. Additionally, since $[tf_1(t)]' =\frac{2(1 - t^2)}{\pi (t^2+1)^2} < 0$ for $t>1$, and $S(t) = \frac{t^2+1}{1 - t^2}$ is increasing on this interval, $\sup_{t>1} S(t) = \lim_{t \to \infty} S(t) = -1$. Consequently, applying \eqref{eq: B'(x) bounds} yields:
\begin{align}
    - \frac{\alpha}{(1+(\alpha x)^2)(\frac{\pi}{2} -\arctan(\alpha x))} \beta(x)\leq B'(x) \leq \frac{\alpha}{(1+(\alpha x)^2)(\frac{\pi}{2} -\arctan(\alpha x))} \beta(x). \nonumber
\end{align}
Now, we select the form of the acceleration function $\beta(x)$ such that $\beta'(x) =\frac{\alpha}{(1+(\alpha x)^2)(\frac{\pi}{2} -\arctan(\alpha x))} \beta(x)$ (that is, equal to $\beta'(x)'s$ upper bound). Hence from \eqref{eq: beta(x) general} with $c^* = 1$, $\beta(x) =  \gamma\frac{1}{1 - \frac{2}{\pi}\arctan(\alpha x)}$ with $\gamma >0$. From Theorem \ref{thm: half-cauchy} in Appendix \ref{app: thms}, since $\phi(x) = xh_0(x) = x\frac{\alpha}{(1+(\alpha x)^2)(\frac{\pi}{2} -\arctan(\alpha x))}$ is strictly increasing on $x\in (0,\infty)$, by Theorem \ref{thm: ABC}, it follows that in order for condition \eqref{eq: tau bounds} to hold, we must have $\tau \in [0, 1]$. Accordingly, we select our acceleration function as:
\begin{align}
    \beta(x) =   \gamma\frac{1}{1- \frac{2}{\pi}\arctan(\alpha \tau x)},\quad x > 0, \nonumber
\end{align}
where $\alpha, \gamma> 0$ with $\tau \in [0, 1]$.

\section{Gamma}
We restrict the survival function forms to the Gamma family as such:
\begin{align}
    \bar{F}_0(x) &= P(X>x) = 1 - \frac{1}{\Gamma(\lambda)}\gamma^*(\lambda, \alpha x) = \frac{\Gamma(\lambda, \alpha x)}{\Gamma(\lambda)}, \quad x> 0,\nonumber
\end{align}
where $\alpha, \lambda >0$ and denoting $\gamma^*$ as the lower incomplete gamma function. Now for each $x>0$:
\begin{align}
    \bar{F}_1(\beta(x) y) &= P(Y>y \mid X>x) =  1 - \frac{1}{\Gamma(\nu)}\gamma^*(\nu, \beta(x) y) = \frac{\Gamma(\nu, \beta(x) y)}{\Gamma(\nu)}, \quad y>0, \nonumber
\end{align}
where $\beta(x), \nu > 0$ for all $x>0 $ to ensure a valid survival function. The corresponding joint survival function will be:
\begin{align}
    \bar{F}_0(x) \bar{F}_1(\beta(x) y) &= P(X>x,Y>y) = \frac{\Gamma(\lambda, \alpha x)\Gamma(\nu, \beta(x)y)}{\Gamma(\lambda) \Gamma(\nu)}, \quad x, y>0. \label{eq: Gamma joint survival}
\end{align}
Accordingly, $f_X(x) = \frac{\alpha^\lambda}{\Gamma(\lambda)} x^{\lambda - 1}e^{-\alpha x}$ for $x> 0$ and $f_Y(y) = \frac{\beta(0)^\nu}{\Gamma(\nu)} y^{\nu - 1}e^{-\beta(0) y}$ for $y> 0$. Hence, $X\sim Gamma(\alpha, \lambda)$ and $Y\sim Gamma(\beta(0), \nu)$.\\

Now since $[tf_1(t)]' = \frac{t^{\nu - 1}(\nu - t)e^{-t}}{\Gamma(\nu)} > 0 $ for $t<\nu$, and $S(t) = \frac{1}{\nu - t}$ is increasing on this interval, $\inf_{t<\nu} S(t) = \lim_{t \to 0^+} S(t) = c^*= \frac{1}{\nu}$. Additionally, since $[tf_1(t)]' =\frac{t^{\nu - 1}(\nu - t)e^{-t}}{\Gamma(\nu)} < 0$ for $t>\nu$, and $S(t) =\frac{1}{\nu - t}$ is increasing on this interval, $\sup_{t>\nu} S(t) = \lim_{t \to \infty} S(t) = 0$. Consequently, applying \eqref{eq: B'(x) bounds} yields:
\begin{align}
    0 \leq B'(x) \leq \frac{1}{\nu} \frac{\alpha^\lambda}{\Gamma(\lambda, \alpha x)}x^{\lambda -1} e^{-\alpha x} \beta(x). \nonumber
\end{align}
Now, we select the form of the acceleration function $\beta(x)$ such that $\beta'(x) =\frac{1}{\nu} \frac{\alpha^\lambda}{\Gamma(\lambda, \alpha x)}x^{\lambda -1} e^{-\alpha x}$ (that is, equal to $\beta'(x)'s$ upper bound). Hence from \eqref{eq: beta(x) general} with $c^* = \frac{1}{\nu}$, $\beta(x) =  \gamma \left[\frac{\Gamma(\lambda)}{\Gamma(\lambda, \alpha x)}\right]^{\frac{1}{\nu}}$ with $\gamma >0$. From Theorem \ref{thm: gamma} in Appendix \ref{app: thms}, since $\phi(x) = xh_0(x) = x \frac{\alpha^\lambda}{\Gamma(\lambda, \alpha x)}x^{\lambda -1} e^{-\alpha x}$ is strictly increasing on $x\in (0,\infty)$, by Theorem \ref{thm: ABC}, it follows that in order for condition \eqref{eq: tau bounds} to hold, we must have $\tau \in [0, 1]$. Accordingly, we select our acceleration function as:
\begin{align}
    \beta(x) =    \gamma \left[\frac{\Gamma(\lambda)}{\Gamma(\lambda, \alpha \tau x)}\right]^{\frac{1}{\nu}},\quad x > 0,  \nonumber
\end{align}
where $\alpha, \lambda, \gamma, \nu  > 0$ with $\tau \in [0, 1]$.

\section{Simulation}
\subsection{Method}
To generate dependent pairs $(X,Y)$ for simulation, we employ a copula-based simulation approach for the method of moments estimation and a Metropolis-Hastings (MH) algorithm for maximum likelihood estimation. m.m.e's Rely exclusively on a limited set of moments; such as means, variances, and covariances - being such, MH samples may fail to reproduce these accurately (unless the Markov chains are extremely long and well-mixed) leading to poor moment-based estimates. Copula-based simulation provides a more effective solution in this context, by generating synthetic $(X,Y)$ pairs that preserve the specified marginal distributions and dependence structure, ensuring that simulated moments are consistent with their theoretical counterparts. In contrast, copula-based simulation is unsuitable for likelihood-based estimation because it alters the joint density of 
$(X,Y)$, whereas m.l.e's are defined relative to the true joint distribution. Using copula-generated samples for maximum likelihood estimation would therefore imply optimizing a likelihood corresponding to a different model.\\
 
Additionally, we note that since we were able to derive $F_{Y\mid X=x}(y)$ for the exponential model, we may merely utilise the inverse transform method to simulate a $Y$ conditioned on $X$, given $X \sim Exp(\alpha)$ and $Y\sim Exp(\beta(0))$. Specifically, making use of \eqref{eq: Fygivx exp}, and conditional on a draw of $X = x$, we may obtain: 
\begin{align}
    Y= -\frac{\tau W\left( \frac{(U-1)e^{-\frac{1}{\tau}}}{\tau}\right) +1}{\gamma \tau e^{\alpha\tau x}} \nonumber
\end{align}
where $W$ is the lower branch of the Lambert W function and $U\sim Uniform(0, 1)$.\\

Furthermore, we conduct simulations based on $10000$\footnote{For the half-Cauchy and gamma models, only $1000$ datasets were used, as the computation of the m.l.e's required substantially more time when employing R’s \texttt{optim} function with the \texttt{L-BFGS-B} algorithm.}
data sets of size $n = 1000$, $5000$ and $10000$ for fixed parameter values: $\alpha = 1, \gamma =2, \lambda =3, \nu = 4$ and $\tau = 0.5$.

\subsubsection{Method of moments}
The copula preserves the marginal 
distributions of $X$ and $Y$ while introducing dependence through $\rho_{\text{latent}}$. 
Specifically, we use a Student-$t$ copula to capture both tail dependence 
and negative correlations. Here, we denote $F_X(x) = 1 - \bar{F}_0(x)$ and 
$F_Y(y) = 1 - \bar{F}_1(\beta(0)y)$.\\

\noindent
\textbf{Step 1:} Construct a $t$-copula with correlation $\rho_{\text{latent}}$ 
and degrees of freedom $v_c$. Formally, for $U = F_X(X)$ and $V = F_Y(Y)$, 
we have $(U,V) \sim C_{v_c,\rho_{\text{latent}}}$, where the copula distribution function is:
\[
C_{v_c,\rho_{\text{latent}}}(u,v) 
= t_{2,v_c}\!\big(t^{-1}_{v_c}(u),\,t^{-1}_{v_c}(v);\,\rho_{\text{latent}}\big), 
\quad (u,v)\in[0,1]^2,
\]
with $t^{-1}_{v_c}$ the quantile function of a univariate Student-$t(v_c)$ and 
$t_{2,v_c}$ the bivariate Student-$t$ distribution function with correlation $\rho_{\text{latent}}$. 
From Sklar's Theorem, the joint density of $(X,Y)$ with marginal distributions $F_X$ and $F_Y$ is:
\[
f_{X,Y}(x,y) 
= c_{v_c,\rho_{\text{latent}}}\!\big(F_X(x),F_Y(y)\big)\, f_X(x)\, f_Y(y),
\]
where the copula density is $c_{v_c,\rho_{\text{latent}}}(u,v) 
= \frac{\partial^2}{\partial u \,\partial v} C_{v_c,\rho_{\text{latent}}}(u,v)$.\\

\noindent
\textbf{Step 2:} Map the model Pearson correlation $\rho_{\text{target}}$ into the copula correlation parameter $\rho_{\text{latent}}$ that drives dependence 
in the $t$-copula. This requires re-expressing the correlation in terms of 
$c_{v_c, \rho_{\text{latent}}}(u, v)$ and solving numerically to obtain $\rho_{\text{latent}}$:
\[
\rho(X,Y) = \frac{{Cov}(X,Y)}{\sqrt{{Var}(X)\,{Var}(Y)}} 
= \frac{\int_0^1 \int_0^1 F_X^{-1}(u)\,F_Y^{-1}(v)\, c_{v_c, \rho_{\text{latent}}}(u,v)\,du\,dv 
- {E}(X){E}(Y)}
{\sqrt{{Var}(X)\,{Var}(Y)}},
\]
with ${E}(X) = \int_0^1 F_X^{-1}(u)\,du$, 
${E}(Y) = \int_0^1 F_Y^{-1}(v)\,dv$, and 
${Var}(X) = {E}(X^2)-{E}(X)^2$.\\

\noindent
\textbf{Step 3:} Simulate $n$ pairs $(U,V)$ from $(U,V) \sim C_{v_c,\rho_{\text{latent}}}$, 
then transform back via the inverse distributions to obtain 
$X = F_X^{-1}(U)$ and $Y = F_Y^{-1}(V)$.\\

Furthermore, Appendix \ref{app: mmes} provides a guide on how to obtain m.m.e's for the models (noting that moments were not provided for the half-Cauchy and gamma models, we exclude m.m.e's for said models).

\subsubsection{Maximum Likelihood Estimation}



We utilise the MH algorithm to sample from the joint density of the models considered in this study, where the target distribution is given by $\pi(x,y)=f_{X,Y}(x,y)$, noting that $X\sim Family(\alpha, \lambda)$ and $Y \sim Family(\beta(0), \nu)$, with $f_{X,Y}(x,y)$ given in Appendix \ref{app: joint lik}.\\

\noindent
\textbf{Step 1 (Initialization):} Choose any starting point \((x^{(0)},y^{(0)})\) in the support of \(f_{X,Y}\). Set \(j=0\). \\

\noindent
\textbf{Step 2 (Proposal):} Given the current state \((x^{(j)},y^{(j)})\), draw a candidate
\[
(x^*,y^*) \sim q(\,\cdot\,\mid x^{(j)},y^{(j)}).
\]
The study utilises an independence proposal:
\[
q(x^*,y^*\mid x^{(j)},y^{(j)}) = q(x^*,y^*)=f_X(x^*)\,f_Y(y^*).
\] \\
\noindent
\textbf{Step 3 (Acceptance ratio $\alpha_a$):} Compute:
\begin{align}
\alpha_a &= \min\left(
\frac{f_{X,Y}(x^*,y^*)\,q(x^{(j)},y^{(j)}\mid x^*,y^*)}{f_{X,Y}(x^{(j)},y^{(j)})\,q(x^*,y^*\mid x^{(j)},y^{(j)})}, 1 \right) \nonumber \\
& = \min \left(\frac{f_{X,Y}(x^*,y^*)}{f_X(x^*)\,f_Y(y^*)}
\bigg/
\frac{f_{X,Y}(x^{(j)},y^{(j)})}{f_X(x^{(j)})\,f_Y(y^{(j)})}, 1\right). \nonumber
\end{align}
Equivalently, using logs for numerical stability,
\begin{align}
\log(\alpha_a) 
&= \min\Bigg\{ 
\Big[\,\log\!\big(f_{X,Y}(x^*, y^*)\big)
      - \log\!\big(f_X(x^*)\big)
      - \log\!\big(f_Y(y^*)\big)\,\Big] \nonumber \\
&\quad\quad -
\Big[\,\log\!\big(f_{X,Y}(x^{(j)}, y^{(j)})\big)
      - \log\!\big(f_X(x^{(j)})\big)
      - \log\!\big(f_Y(y^{(j)})\big)\,\Big],
\, 0 \Bigg\}. \nonumber
\end{align}
Subsequently, $(x^*, y^*)$ is accepted as the new pair in the Markov chain under the following acceptance criterion:
\begin{align}
(x^{(j+1)}, y^{(j+1)}) =
\begin{cases}
(x^*, y^*), & \text{if } U < \alpha_a \\
(x^{(j)}, y^{(j)}), & \text{otherwise.} 
\end{cases}\nonumber
\end{align}
where $U\sim Uniform(0, 1)$.\\

Furthermore, Appendix \ref{app: joint lik} presents the log-likelihoods for the study's models - m.l.e's are obtained numerically as $\operatorname{argmax}_{\boldsymbol{\theta}} \ell (\boldsymbol{\theta})$ where $\boldsymbol{\theta} = \left[\alpha, \lambda, \gamma, \nu, \tau\right]'$.

\subsection{Simulation Data}
The corresponding m.m.e and m.l.e results are reported in Tables~\ref{Table: Simulations for exp model} - \ref{Table: Simulations for Gamma model}, together with $95\%$ confidence intervals computed as $\hat{\theta} \pm Z_{\alpha/2}\text{S.E.}(\hat{\theta})$, where $\hat{\theta}$ denotes the point estimator\footnote{The reported standard errors (S.E.) are calculated as the empirical standard deviation of the parameter estimates across simulated datasets, representing the sampling variability of the estimator.}. In addition, we report Pearson correlation coefficients: $\text{PC}^{\text{C}}$ refers to correlations from datasets generated via the copula, while $\text{PC}^{\text{MH}}$ refers to correlations from datasets generated using the MH algorithm.\\

As expected, the standard errors of both m.m.e's and m.l.e's \footnote{m.m.e's (where applicable) were used as initial values when using R's \texttt{optim} function.} decrease as the sample size $n$ increases, yielding narrower $95\%$ confidence intervals. We refrain from directly comparing the accuracy of the two estimation methods, since they were applied to different simulated datasets. A limitation of the simulation is the relative inaccuracy of Pearson correlations obtained for datasets derived from the copula in the log-logistic model (see $\text{PC}^{\text{C}}$ in Table~\ref{Table: Simulations for log-logistic model}), as well as the corresponding inaccuracy of the Pearson correlations obtained for datasets generated by the MH algorithm in the Lomax model (see $\text{PC}^{\text{MH}}$ in Table~\ref{Table: Simulations for Lomax model}). Finally, the simulations indicate that the Lomax model exhibits some instability in parameter estimation (Table~\ref{Table: Simulations for Lomax model}); reliable estimation appears to require substantially larger sample sizes.\\

Figure \ref{fig: sims gif} displays the bivariate densities of the models used in the study given the fixed parameter set.

\begin{table}[H] \centering 
	\begin{tabular}{@{\extracolsep{1pt}} ccccccccc} 
		\\[-1.8ex]\hline 
		\hline \\[-1.8ex] 
		n & Parameter & m.m.e & m.l.e & SE(m.m.e) & SE(m.l.e) & $95\%$CI (m.m.e) & $95\%$CI (m.l.e) & PC \\ 
		\hline \\[-1.8ex] 
 & $\alpha$ & 1.001 & 1.000 & 0.032 & 0.031 & $(0.938,\,1.063)$ & $(0.939,\,1.061)$ &  \\ 
    1000 & $\gamma$ & 2.001 & 2.000 & 0.063 & 0.063 & $(1.877,\,2.125)$ & $(1.878,\,2.123)$ &  -0.334 \\ 
     & $\tau$   & 0.499 & 0.501 & 0.045 & 0.028 &  $(0.410,\,0.588)$ & $(0.445,\,0.557)$ &  \\ 
     & $\rho$   & -0.333 & -0.334 & 0.020 & 0.013 & $(-0.373,\,-0.294)$ & $(-0.359,\,-0.309)$ & 
\\[1ex]
\hline 
 & $\alpha$ & 1.001 & 1.000 & 0.014 & 0.0138 & $(0.973,\,1.028)$ & $(0.973,\,1.027)$ & \\
    5000 & $\gamma$ & 2.001 & 2.000 & 0.028 & 0.0279 & $(1.945,\,2.056)$ & $(1.945,\,2.055)$ &  -0.334\\
     & $\tau$   & 0.500 & 0.500 & 0.020 & 0.0126 & $(0.460,\,0.539)$ & $(0.475,\,0.525)$ & \\
     & $\rho$   & -0.333 & -0.333 & 0.009 & 0.0056 & $(-0.351,\,-0.316)$ & $(-0.344,\,-0.322)$ & 
\\[1ex]
\hline  
 & $\alpha$ & 1.000 & 1.000 & 0.0100 & 0.0097 & $(0.980,\,1.020)$ & $(0.981,\,1.019)$ & \\
     10000 & $\gamma$ & 2.000 & 2.000 & 0.0203 & 0.0199 & $(1.960,\,2.040)$ & $(1.961,\,2.039)$ & -0.334 \\
      & $\tau$   & 0.500 & 0.500 & 0.0143 & 0.0090 & $(0.472,\,0.528)$ & $(0.482,\,0.518)$ & \\
      & $\rho$   & -0.333 & -0.333 & 0.0063 & 0.0040 & $(-0.346,\,-0.321)$ & $(-0.341,\,-0.326)$ & \\
		\hline
	\end{tabular} 
    \caption{Simulations for exponential model ($\alpha = 1$, $\gamma =2$,  $\tau = 0.5$ with $\rho = -0.333$).} 
    \label{Table: Simulations for exp model} 
\end{table} 

\begin{table}[H] \centering 
	\begin{tabular}{@{\extracolsep{1pt}} cccccccccc} 
		\\[-1.8ex]\hline 
		\hline \\[-1.8ex] 
		n & Parameter & m.m.e & m.l.e & SE(m.m.e) & SE(m.l.e) & $95\%$CI (m.m.e) & $95\%$CI (m.l.e) & PC$^{\text{C}}$ & PC$^{\text{MH}}$\\ 
		\hline \\[-1.8ex] 
        & $\alpha$ & 0.843 & 0.848 & 0.199 & 0.268 & $(0.452,\,1.234)$ & $(0.323,\,1.374)$ & \\
     & $\gamma$ & 1.795 & 1.889 & 0.495 & 0.622 & $(0.824,\,2.766)$ & $(0.670,\,3.109)$ & \\
     1000 & $\lambda$ & 3.383 & 3.406 & 0.591 & 0.888 & $(2.224,\,4.542)$ & $(1.664,\, 5.146)$ &-0.212& -0.236 \\
     & $\nu$   & 4.349 & 4.178 & 0.945 & 1.120 & $(2.496,\,6.201)$ & $(1.984,\,6.373)$ & \\
     & $\tau$  & 0.559 & 0.481 & 0.145 & 0.090 & $(0.276,\,0.843)$ & $(0.304,\,0.657)$ & \\
     & $\rho$  & -0.212 & -0.195 & 0.044 & 0.041 & $(-0.298,\,-0.125)$ & $(-0.276,\,-0.113)$ & 
\\[1ex]
\hline 
 & $\alpha$ & 0.911 & 0.926 & 0.148 & 0.157 & $(0.622,\,1.201)$ & $(0.618,\,1.235)$ & \\
     & $\gamma$ & 1.906 & 1.963 & 0.315 & 0.332 & $(1.289,\,2.524)$ & $(1.313,\,2.613)$ & \\
    5000 & $\lambda$ & 3.198 & 3.202 & 0.323 & 0.465 & $(2.565,\,3.832)$ & $(2.291,\,4.113)$ & -0.1987 & -0.224\\
     & $\nu$   & 4.153 & 4.062 & 0.456 & 0.440 & $(3.259,\,5.046)$ & $(3.200,\,4.924)$ & \\
     & $\tau$  & 0.522 & 0.491 & 0.089 & 0.051 & $(0.347,\,0.696)$ & $(0.391,\,0.592)$ & \\
     & $\rho$  & -0.199 & -0.192 & 0.030 & 0.023 & $(-0.257,\,-0.141)$ & $(-0.237,\,-0.148)$ & 
\\[1ex]
\hline 
 & $\alpha$ & 0.931 & 0.946 & 0.131 & 0.125 & $(0.674,\,1.189)$ & $(0.700,\,1.192)$ & \\
       & $\gamma$ & 1.935 & 1.974 & 0.245 & 0.227 & $(1.454,\,2.416)$ & $(1.529,\,2.419)$ & \\
       10000 & $\lambda$ & 3.149 & 3.143 & 0.266 & 0.349 & $(2.628,\,3.670)$ & $(2.458,\,3.828)$ & -0.1948 & -0.220\\
       & $\nu$   & 4.102 & 4.038 & 0.349 & 0.317 & $(3.419,\,4.785)$ & $(3.417,\,4.658)$ & \\
       & $\tau$  & 0.512 & 0.494 & 0.073 & 0.040 & $(0.368,\,0.656)$ & $(0.416,\,0.572)$ & \\
       & $\rho$  & -0.195 & -0.191 & 0.025 & 0.017 & $(-0.243,\,-0.146)$ & $(-0.225,\,-0.157)$ & \\
		\hline
	\end{tabular} 
    \caption{Simulations for Lomax model ($\alpha = 1$, $\gamma =2$, $\lambda = 3$, $\nu = 4$, $\tau = 0.5$ with $\rho = -0.187$).} 
    \label{Table: Simulations for Lomax model} 
\end{table}

\begin{table}[H] \centering 
	\begin{tabular}{@{\extracolsep{1pt}} cccccccccc} 
		\\[-1.8ex]\hline 
		\hline \\[-1.8ex] 
		n & Parameter & m.m.e & m.l.e & SE(m.m.e) & SE(m.l.e) & $95\%$CI (m.m.e) & $95\%$CI (m.l.e) & PC$^{\text{C}}$ & PC$^{\text{MH}}$\\ 
		\hline \\[-1.8ex] 
 & $\alpha$ & 1.000 & 1.000 & 0.011 & 0.012 & $(0.978,\,1.022)$ & $(0.976,\,1.025)$ & \\
     & $\gamma$ & 2.000 & 2.000 & 0.017 & 0.018 & $(1.968,\,2.033)$ & $(1.964,\,2.036)$ & \\
     1000 & $\lambda$ & 3.003 & 3.005 & 0.076 & 0.083 & $(2.854,\,3.151)$ & $(2.842,\,3.167)$ & -0.101 & -0.100\\
     & $\nu$   & 4.004 & 4.004 & 0.103 & 0.109 & $(3.802,\,4.205)$ & $(3.790,\,4.218)$ & \\
     & $\tau$  & 0.501 & 0.501 & 0.075 & 0.050 & $(0.354,\,0.648)$ & $(0.402,\,0.600)$ & \\
     & $\rho$  & -0.101 & -0.100 & 0.037 & 0.029 & $(-0.173,\,-0.029)$ & $(-0.157,\,-0.044)$ & 
\\[1ex]
\hline 
 & $\alpha$ & 1.000 & 1.000 & 0.0050 & 0.0056 & $(0.990,\,1.010)$ & $(0.989,\,1.011)$ & \\
     & $\gamma$ & 2.000 & 2.000 & 0.0075 & 0.0082 & $(1.985,\,2.015)$ & $(1.984,\,2.016)$ & \\
     5000 & $\lambda$ & 3.001 & 3.001 & 0.034 & 0.038 & $(2.934,\,3.067)$ & $(2.927,\,3.074)$ & -0.100 & -0.100 \\
     & $\nu$   & 4.001 & 4.000 & 0.047 & 0.049 & $(3.909,\,4.092)$ & $(3.904,\,4.096)$ & \\
     & $\tau$  & 0.500 & 0.500 & 0.028 & 0.021 & $(0.444,\,0.556)$ & $(0.459,\,0.541)$ & \\
     & $\rho$  & -0.100 & -0.100 & 0.016 & 0.013 & $(-0.132,\,-0.068)$ & $(-0.125,\,-0.075)$ &
\\[1ex]
\hline 
 & $\alpha$ & 1.000 & 1.000 & 0.0050 & 0.0039 & $(0.990,\,1.010)$ & $(0.992,\,1.008)$ & \\
       & $\gamma$ & 2.000 & 2.000 & 0.0075 & 0.0058 & $(1.985,\,2.015)$ & $(1.989,\,2.011)$ & \\
      10000 & $\lambda$ & 3.001 & 3.000 & 0.0337 & 0.0265 & $(2.934,\,3.067)$ & $(2.948,\,3.052)$ & -0.100 & -0.100\\
       & $\nu$   & 4.001 & 4.000 & 0.0467 & 0.0344 & $(3.909,\,4.092)$ & $(3.933,\,4.067)$ & \\
       & $\tau$  & 0.500 & 0.500 & 0.0284 & 0.0147 & $(0.444,\,0.556)$ & $(0.471,\,0.529)$ & \\
       & $\rho$  & -0.100 & -0.100 & 0.0163 & 0.0089 & $(-0.132,\,-0.068)$ & $(-0.118,\,-0.083)$ & \\
		\hline
	\end{tabular} 
    \caption{Simulations for Weibull model ($\alpha = 1$, $\gamma =2$, $\lambda = 3$, $\nu = 4$, $\tau = 0.5$ with $\rho = -0.100$).} 
    \label{Table: Simulations for Weibull model} 
\end{table}

\begin{table}[H] \centering 
	\begin{tabular}{@{\extracolsep{1pt}} cccccccccc} 
		\\[-1.8ex]\hline 
		\hline \\[-1.8ex] 
		n & Parameter & m.m.e & m.l.e & SE(m.m.e) & SE(m.l.e) & $95\%$CI (m.m.e) & $95\%$CI (m.l.e) & PC$^{\text{C}}$ & PC$^{\text{MH}}$\\ 
		\hline \\[-1.8ex] 
& $\alpha$ & 0.987 & 1.001 & 0.139 & 0.024 & $(0.715,\,1.259)$ & $(0.954,\,1.048)$ & \\
     & $\gamma$ & 1.996 & 1.998 & 0.034 & 0.033 & $(1.929,\,2.063)$ & $(1.934,\,2.063)$ & \\
    1000 & $\lambda$ & 3.146 & 3.022 & 0.228 & 0.125 & $(2.700,\,3.593)$ & $(2.778,\,3.266)$ & -0.150  &  -0.117 \\
     & $\nu$   & 4.084 & 4.013 & 0.238 & 0.125 & $(3.618,\,4.549)$ & $(3.767,\,4.259)$ & \\
     & $\tau$  & 0.573 & 0.488 & 0.121 & 0.079 & $(0.336,\,0.810)$ & $(0.332,\,0.644)$ & \\
     & $\rho$  & -0.150 & -0.119 & 0.045 & 0.030 & $(-0.239,\,-0.061)$ & $(-0.178,\,-0.061)$ & 
\\[1ex]
\hline 
 & $\alpha$ & 0.992 & 1.001 & 0.102 & 0.012 & $(0.791,\,1.192)$ & $(0.978,\,1.024)$ & \\
    & $\gamma$ & 1.998 & 1.999 & 0.019 & 0.016 & $(1.961,\,2.034)$ & $(1.968,\,2.030)$ & \\
   5000  & $\lambda$ & 3.080 & 3.011 & 0.156 & 0.068 & $(2.774,\,3.386)$ & $(2.878,\,3.144)$ & -0.137 & -0.122 \\
    & $\nu$   & 4.036 & 4.006 & 0.141 & 0.066 & $(3.760,\,4.312)$ & $(3.876,\,4.136)$ & \\
    & $\tau$  & 0.537 & 0.494 & 0.081 & 0.037 & $(0.378,\,0.696)$ & $(0.421,\,0.567)$ & \\
    & $\rho$  & -0.137 & -0.122 & 0.029 & 0.014 & $(-0.193,\,-0.080)$ & $(-0.150,\,-0.094)$ & 
\\[1ex]
\hline 
 & $\alpha$ & 0.993 & 1.000 & 0.098 & 0.008 & $(0.802,\,1.185)$ & $(0.984,\,1.016)$ & \\
       & $\gamma$ & 1.998 & 2.000 & 0.013 & 0.011 & $(1.972,\,2.024)$ & $(1.978,\,2.021)$ & \\
       10000 & $\lambda$ & 3.060 & 3.007 & 0.137 & 0.052 & $(2.792,\,3.327)$ & $(2.906,\,3.108)$ & -0.134 & -0.124\\
       & $\nu$   & 4.023 & 4.004 & 0.106 & 0.049 & $(3.816,\,4.231)$ & $(3.907,\,4.100)$ & \\
       & $\tau$  & 0.528 & 0.496 & 0.068 & 0.028 & $(0.395,\,0.662)$ & $(0.442,\,0.550)$ & \\
       & $\rho$  & -0.134 & -0.123 & 0.024 & 0.011 & $(-0.181,\,-0.088)$ & $(-0.144,\,-0.102)$ & \\
		\hline
	\end{tabular} 
    \caption{Simulations for log-logistic model ($\alpha = 1$, $\gamma =2$, $\lambda = 3$, $\nu = 4$, $\tau = 0.5$ with $\rho = -0.125$).} 
    \label{Table: Simulations for log-logistic model} 
\end{table}

\begin{table}[H] \centering 
	\begin{tabular}{@{\extracolsep{1pt}} cccccc} 
		\\[-1.8ex]\hline 
		\hline \\[-1.8ex] 
		n & Parameter & m.l.e & SE(m.l.e) & $95\%$CI (m.l.e)  & PC$^{\text{MH}}$\\ 
		\hline \\[-1.8ex] 
 & $\alpha$ & 1.009 & 0.070 & $(0.872,\,1.146)$ & \\
    1000 & $\gamma$ & 1.987 & 0.116 & $(1.759,\,2.215)$ & -0.032 \\
     & $\tau$   & 0.488 & 0.069 & $(0.352,\,0.624)$ &
\\[1ex]
\hline 
 & $\alpha$ & 1.003 & 0.041 & $(0.922,\,1.085)$ & \\
     5000 & $\gamma$ & 1.996 & 0.073 & $(1.853,\,2.140)$ &  -0.014 \\
     & $\tau$   & 0.496 & 0.040 & $(0.418,\,0.573)$ & 
\\[1ex]
\hline 
 & $\alpha$ & 1.002 & 0.031 & $(0.941,\,1.062)$ & \\
       10000 & $\gamma$ & 1.998 & 0.051 & $(1.898,\,2.098)$ & -0.009 \\
       & $\tau$   & 0.497 & 0.029 & $(0.440,\,0.555)$ & \\
		\hline
	\end{tabular} 
    \caption{Simulations for Half-Cauchy model ($\alpha = 1$, $\gamma =2$, $\tau = 0.5$).} 
    \label{Table: Simulations for Half-Cauchy model} 
\end{table}

\begin{table}[H] \centering 
	\begin{tabular}{@{\extracolsep{1pt}} cccccc} 
		\\[-1.8ex]\hline 
		\hline \\[-1.8ex] 
		n & Parameter & m.l.e & SE(m.l.e) & $95\%$CI (m.l.e)  & PC$^{\text{MH}}$\\ 
		\hline \\[-1.8ex] 
&$\alpha$ & 1.007 & 0.058 & $(0.894,\,1.120)$ \\ 
&$\gamma$ & 2.007 & 0.105 & $(1.800,\,2.213)$ \\ 
1000 &$\lambda$ & 3.016 & 0.146 & $(2.729,\,3.303)$ & -0.122 \\ 
&$\nu$     & 4.012 & 0.197 & $(3.626,\,4.398)$ \\ 
&$\tau$    & 0.497 & 0.069 & $(0.362,\,0.633)$ 
\\[1ex]
\hline 
    &$\alpha$ & 1.003 & 0.026 & $(0.951,\,1.054)$ \\ 
    &$\gamma$ & 2.002 & 0.045 & $(1.913,\,2.090)$ \\ 
    5000&$\lambda$ & 3.005 & 0.066 & $(2.876,\,3.134)$ & -0.122\\ 
    &$\nu$     & 4.004 & 0.085 & $(3.838,\,4.169)$ \\ 
    &$\tau$    & 0.499 & 0.032 & $(0.437,\,0.561)$ 
\\[1ex]
\hline 
    &$\alpha$ & 1.003 & 0.019 & $(0.966,\,1.039)$ \\ 
    &$\gamma$ & 2.000 & 0.033 & $(1.936,\,2.064)$ \\ 
    10000&$\lambda$ & 3.006 & 0.048 & $(2.913,\,3.099)$ & -0.123  \\ 
    &$\nu$     & 3.998 & 0.062 & $(3.877,\,4.120)$ \\ 
    &$\tau$    & 0.499 & 0.022 & $(0.455,\,0.543)$ \\ 
		\hline
	\end{tabular} 
    \caption{Simulations for gamma model ($\alpha = 1$, $\gamma =2$, $\lambda = 3$, $\nu = 4$ and $\tau = 0.5$).} 
    \label{Table: Simulations for Gamma model} 
\end{table}

\begin{figure}[H]
    \centering
\animategraphics[controls,autoplay,loop,width= 0.5\textwidth]{1}{sims/}{1}{6}
\caption{Bivariate densities of models given $\alpha = 1, \gamma = 2, \lambda =3, \nu = 4$ and $\tau = 0.5$.}
\label{fig: sims gif}
\end{figure}

\section{Two Small Applications}
The first dataset employs the song metric speechiness as the independent variable 
$X$ and song popularity as the dependent variable $Y$ conditioned on $\{X>x\}$, based on the most popular TikTok songs in 2020 available from \citep{sveta1512020tiktok}\footnote{Song popularity ($Y$) was originally reported on a $0-100$ scale. For analysis, these values were normalized to the unit interval by dividing by $100$, with songs assigned a score of $0$ removed from the dataset.}. The pearson correlation is $-0.322$; which ensures all models in this study may be utilised. Additionally, it is such that $S_1 < M_1^2$ and $S_2 < M_2^2$ for this dataset, hence m.m.e's for the Lomax model are excluded.\\

The second dataset employs the logarithm of the volatility index (VIX) as the independent variable $X$ and the logarithm of the price of gold futures (GC) as the dependent variable $Y$ conditioned on $\{X>x\}$, covering the period from 1 January 2020 to 17 October 2025\footnote{Rows containing missing values were removed - arising primarily from non-trading days (weekends and public holidays) and differences in trading calendars across the VIX and gold futures contracts.}. The data was obtained from \citep{yahoo2025finance}. The pearson correlation is $-0.311$; which ensures all models in this study may be utilised. Again, it is such that $S_1 < M_1^2$ and $S_2 < M_2^2$ for this dataset, hence m.m.e's for the Lomax model are excluded. \\

The corresponding m.m.e's, m.l.e's, and AIC values for both datasets are reported in Table \ref{table: tiktok vix}. An important observation is that the poorest-fitting models tend to be the exponential, Lomax, and half-Cauchy models. This outcome is unsurprising, as these distributions are characterised by non-increasing densities and are therefore unable to adequately capture unimodal, skewed marginals - a feature evident in both the $X$ and $Y$ variables across both datasets. This finding reinforces the necessity of conducting goodness-of-fit tests on the datasets' $X$ and $Y$ variables prior to fitting the joint models proposed in this study. In other words, given that our models assume $X\sim Family(\alpha, \lambda)$ and $Y \sim Family (\beta(0), \nu)$, it is prudent to first assess whether the $X$ and $Y$ variables from the datasets are indeed compatible with the assumed distributional families.\\

A clear example of poor fit arises from the Lomax model’s parameter estimates for the VIX/GC dataset: the m.l.e's collapse to boundary-degenerate parameter values. Other examples of poor model fit arise from the estimates of the dependence parameter $\tau$ for the exponential and half-Cauchy models on the VIX/GC dataset, where the estimate of $\tau$ approaches the boundary value of either $0$ or $1$. In these cases, the corresponding correlation estimates reach the models' theoretical correlation limits.\\

Contour plots of the estimated bivariate density functions for each model, applied to both datasets, are presented in Figure \ref{fig: application gif}, using m.l.e's from Table \ref{table: tiktok vix}.

\begin{table}[H]
    \centering
    \begin{tabular}{cccccccc}
       & & \multicolumn{3}{c}{TikTok} & \multicolumn{3}{c}{VIX/GC} \\
         \cmidrule(lr){3-5}  \cmidrule(lr){6-8} 
             \hline
     Model & Parameter & m.m.e & m.l.e  & AIC & m.m.e & m.l.e  & AIC \\
     \hline
    Exponential  &  $\alpha$ & 7.003 & 6.174 & -232.041 & 0.334 & 0.327 & 13847.215\\
    & $\gamma$ & 1.567 & 1.585 & &  0.131 & 0.098 &\\
    & $\tau$ & 0.090 & 0.384 & & 0.001 & 1.000 & \\
    & $\rho$ & -0.083 & -0.278 & & -0.001 & -0.500 &
              \\[1ex]
\hline 
Lomax & $\alpha$ &- & 1.477 & -231.381 & -&  0.000 & 13851.245 \\
& $\lambda$ & - & 4.565 & & - & $\infty$ &  \\
& $\gamma$ & - & 0.000 & & - & 0 & \\
& $\nu$ &- & $\infty$ & & - & $\infty$ \\
& $\tau$ & - & 0.484 & & - & 1.000\\
& $\rho$ & - & -0.301 & &- & -0.500 &
 \\[1ex]
\hline 
Weibull & $\alpha$ & 6.936 & 6.391 & \textbf{-715.220} & 0.319 & 0.319 & 1356.375\\
& $\lambda$ & 1.024 & 1.123 & & 11.565 & 8.682 &\\
& $\gamma$ & 1.426 & 1.459 & & 0.129 & 0.129 & \\
& $\nu$ & 4.297 & 4.449 & & 44.197 & 30.174 & \\
& $\tau$ & 0.406 & 0.425 & & 0.938 & 0.795\\
& $\rho$ & -0.322 & -0.311 &  &-0.311 & -0.091 &
\\[1ex]
\hline 
Log-logistic & $\alpha$ & 8.835 & 10.917 & -592.039 & 0.336 & 0.336 & \textbf{10.192}  \\
& $\lambda$ & 2.725 & 2.000 & & 17.416 & 15.379 & \\
& $\gamma$ & 1.618 & 1.531 & & 0.131 & 0.132 & \\
& $\nu$ & 7.179 & 5.189 & & 63.541 & 69.948 &\\
& $\tau$ & 0.855 & 0.456 & & 0.954 & 0.988 &  \\
& $\rho$ & -0.322 & 0.000 & & -0.311 & -0.460
\\[1ex]
\hline 
Half-Cauchy & $\alpha$ &- & 9.968 & -129.483 & - & 0.334 & 15032.367 \\
& $\gamma$ & - & 1.576 & & - & 0.114 & \\
& $\tau$ & - & 0.570 & & - & 1.000 &
 \\[1ex]
\hline 
Gamma & $\alpha$ &- & 9.432 & -517.098 & - & 6.670 & 6183.966\\
& $\lambda$ & - & 1.370 &  &- & 20.958 & \\
& $\gamma$ & - & 8.440 &  &-& 2.656 & \\
& $\nu$ &- & 5.383 &  &-& 19.505 & \\
& $\tau$ & - & 0.641 & & - & 0.232 & \\
\hline
    \end{tabular}
    \caption{m.m.e's, m.l.e's and AIC's for models using TikTok and VIX/GC datasets.}
    \label{table: tiktok vix}
\end{table}

\begin{figure}[H]
  \centering
  \begin{minipage}[t]{0.48\textwidth}
    \centering
\animategraphics[controls,autoplay,loop,width=\textwidth]{1}{Application/TikTok/}{1}{6}
    \caption{Contour plots of estimated bivariate densities of models on TikTok dataset (observations in black dots).}
  \end{minipage}
  \hfill
  \begin{minipage}[t]{0.48\textwidth}
\centering    \animategraphics[controls,autoplay,loop,width=\textwidth]{1}{Application/VIXGC/}{1}{6}
    \caption{Contour plots of estimated bivariate densities of models on VIX/GC dataset (observations in black dots).}
  \end{minipage}

  \label{fig: application gif}
\end{figure}

\section{Mixture Models}
This study employed models based on a single-family specification, wherein the survival function forms are restricted such that both marginal distributions follow the same distributional family but may differ in their respective parameters. Alternatively, a heterogeneous-family specification could be adopted, in which $X$ and $Y$ are drawn from distinct distributional families. Such an approach offers greater flexibility in modelling datasets whose marginals exhibit differing degrees of skewness or tail behaviour.

\appendix
\section*{Appendices}
\section{Joint Densities and Log-likelihoods} \label{app: joint lik}
\subsection{Exponential} 
We obtain the joint density function by differentiating the joint survival function \eqref{eq: exp joint survival}:
\[
f_{X,Y}(x,y) 
= \alpha \gamma \,\bigl(\gamma \tau y e^{\alpha \tau x} - \tau + 1 \bigr) 
   e^{\!\left(\alpha x(\tau - 1) - \gamma y e^{\alpha \tau x}\right)},
   \qquad x,y>0.
\]
where $\alpha, \gamma > 0 $ and $\tau \in [0, 1]$. The log-likelihood is thus:
\begin{align}
\ell(\alpha,\gamma,\tau)
&= n \log \alpha + n \log \gamma 
+ \sum_{i=1}^n \log\!\Big( \gamma \tau y_i e^{\alpha \tau x_i} - \tau + 1 \Big) \nonumber \\
&\quad + \alpha (\tau - 1)\sum_{i=1}^n x_i 
- \gamma \sum_{i=1}^n y_i e^{\alpha \tau x_i}. \nonumber
\end{align}
\subsection{Lomax}
We obtain the joint density function by differentiating the joint survival function \eqref{eq: Lomax joint survival}:
\begin{align*}
f_{X,Y}(x,y) 
&= \nu \,(1+\alpha x)^{-3\lambda - 1}\,
     \bigl(1+\gamma y (1+\alpha \tau x)^{\lambda}\bigr)^{-3\nu - 4} \\[6pt]
&\quad \times \Biggl[
   \frac{\alpha \gamma^{2} \lambda \tau y (\nu+1)\,(1+\alpha x)^{2\lambda+1}(1+\alpha \tau x)^{2\lambda}
        \,\bigl(1+\gamma y (1+\alpha \tau x)^{\lambda}\bigr)^{2\nu+2}}
        {1+\alpha \tau x} \\[6pt]
&\qquad - \frac{\alpha \gamma \lambda \tau \,(1+\alpha x)^{2\lambda+1}(1+\alpha \tau x)^{\lambda}
        \,\bigl(1+\gamma y (1+\alpha \tau x)^{\lambda}\bigr)^{2\nu+3}}
        {1+\alpha \tau x} \\[6pt]
&\qquad + \alpha \gamma \lambda \,\bigl((1+\alpha x)^{2}(1+\alpha \tau x)\bigr)^{\lambda}
        \,\bigl(1+\gamma y (1+\alpha \tau x)^{\lambda}\bigr)^{2\nu+3}
   \Biggr], \quad x,y >0,
\end{align*}
where $\alpha, \lambda, \nu,\gamma > 0 $ and $\tau \in [0, 1]$. The log-likelihood is thus:
\begin{align}
\ell(\alpha,\lambda,\nu,\gamma,\tau)
&= n \log \nu 
 - (3\lambda+1)\sum_{i=1}^n \log(1+\alpha x_i) \nonumber \\
&\quad - (3\nu+4)\sum_{i=1}^n \log\!\bigl(1+\gamma y_i (1+\alpha \tau x_i)^{\lambda}\bigr) \nonumber \\
&\quad + \sum_{i=1}^n 
   \log\!\Biggl\{
   \frac{\alpha \gamma^{2} \lambda \tau y_i (\nu+1)\,(1+\alpha x_i)^{2\lambda+1}(1+\alpha \tau x_i)^{2\lambda}}
        {1+\alpha \tau x_i}
        \bigl(1+\gamma y_i (1+\alpha \tau x_i)^{\lambda}\bigr)^{2\nu+2} \nonumber \\
&\qquad\qquad
   - \frac{\alpha \gamma \lambda \tau (1+\alpha x_i)^{2\lambda+1}(1+\alpha \tau x_i)^{\lambda}}
        {1+\alpha \tau x_i}
        \bigl(1+\gamma y_i (1+\alpha \tau x_i)^{\lambda}\bigr)^{2\nu+3} \nonumber \\
&\qquad\qquad
   + \alpha \gamma \lambda \,\bigl((1+\alpha x_i)^{2}(1+\alpha \tau x_i)\bigr)^{\lambda}
        \bigl(1+\gamma y_i (1+\alpha \tau x_i)^{\lambda}\bigr)^{2\nu+3}
   \Biggr\}. \nonumber
\end{align}
\subsection{Weibull}
We obtain the joint density function by differentiating the joint survival function \eqref{eq: Weibull joint survival}:
\[
f_{X, Y}(x,y) = \alpha^{\lambda} \gamma^{\nu} \lambda \nu x^{\lambda - 1} y^{\nu - 1}
\left( \tau^{\lambda} (\gamma y)^{\nu} e^{(\alpha \tau x)^\lambda}
- \tau^{\lambda} + 1 \right)
e^{(\alpha x)^\lambda(\tau^\lambda - 1)
- (\gamma y)^{\nu} e^{(\alpha\tau x)^{\lambda}}}, \quad x, y>0,
\]
where $\alpha, \lambda, \nu,\gamma > 0 $ and $\tau \in [0, 1]$. The log-likelihood is thus:
\begin{align}
\ell(\alpha,\lambda,\nu,\gamma,\tau)
&= n \lambda \log \alpha + n \nu \log \gamma 
   + n \log \lambda + n \log \nu  \nonumber\\
&\quad + (\lambda - 1) \sum_{i=1}^n \log x_i
   + (\nu - 1) \sum_{i=1}^n \log y_i \nonumber \\
&\quad + \sum_{i=1}^n 
   \log\!\Bigl(\tau^{\lambda} (\gamma y_i)^{\nu} e^{(\alpha \tau x_i)^\lambda}
   - \tau^{\lambda} + 1 \Bigr) \nonumber \\
&\quad + (\tau^{\lambda} - 1)\sum_{i=1}^n (\alpha x_i)^\lambda
   - \sum_{i=1}^n (\gamma y_i)^{\nu} e^{(\alpha \tau x_i)^\lambda}. \nonumber
\end{align}
\subsection{Log-logistic}
We obtain the joint density function by differentiating the joint survival function \eqref{eq: Log joint survival}:
\[
\begin{aligned}
f_{X,Y}(x,y)
&= \frac{\lambda \nu\,(\gamma y)^{\nu}}
{x\,y\,\bigl((\alpha x)^{\lambda}+1\bigr)^{2}
\Bigl( (\gamma y)^{\nu}\bigl((\alpha \tau x)^{\lambda}+1\bigr)+1 \Bigr)^{3}}
\\[4pt]
&\quad\times
\Biggl\{
(\alpha x)^{\lambda}\bigl((\alpha \tau x)^{\lambda}+1\bigr)
\Bigl( (\gamma y)^{\nu}\bigl((\alpha \tau x)^{\lambda}+1\bigr)+1 \Bigr)
\\
&\qquad\qquad
+ (\alpha \tau x)^{\lambda}\bigl((\alpha x)^{\lambda}+1\bigr)
\Bigl( (\gamma y)^{\nu}\bigl((\alpha \tau x)^{\lambda}+1\bigr)+1 \Bigr)
\\
&\qquad\qquad
- 2\,(\alpha \tau x)^{\lambda}\bigl((\alpha x)^{\lambda}+1\bigr)
\Biggr\},
\qquad x,y>0,
\end{aligned}
\]

where $\alpha, \lambda, \nu,\gamma > 0 $ and $\tau \in [0, 1]$. Thus the log-likelihood is:
\begin{align}
\ell(\alpha,\lambda,\nu,\gamma,\tau)
&= n\log\lambda + n\log\nu + n\,\nu\log\gamma
   + (\nu-1)\sum_{i=1}^n \log y_i
   - \sum_{i=1}^n \log x_i \nonumber \\
&\quad - 2 \sum_{i=1}^n \log\!\big( (\alpha x_i)^{\lambda} + 1 \big)
      - 3 \sum_{i=1}^n \log\!\big( (\gamma y_i)^{\nu}\big( (\alpha \tau x_i)^{\lambda} + 1 \big) + 1 \big) \nonumber \\
&\quad + \sum_{i=1}^n \log\!\Big\{
(\alpha x_i)^{\lambda}\big( (\alpha \tau x_i)^{\lambda} + 1 \big)
     \big( (\gamma y_i)^{\nu}\big( (\alpha \tau x_i)^{\lambda} + 1 \big) + 1 \big) \nonumber \\
&\hspace{3.2cm}
+ (\alpha \tau x_i)^{\lambda}\big( (\alpha x_i)^{\lambda} + 1 \big)
     \big( (\gamma y_i)^{\nu}\big( (\alpha \tau x_i)^{\lambda} + 1 \big) + 1 \big) \nonumber \\
&\hspace{3.2cm}
- 2\,(\alpha \tau x_i)^{\lambda}\big( (\alpha x_i)^{\lambda} + 1 \big)
\Big\}. \nonumber
\end{align}
\subsection{Half-Cauchy}
We obtain the joint density function by differentiating the joint survival function \eqref{eq: Cauchy joint survival}:

\begin{align}
f_{X,Y}(x,y)
&=  
\frac{4 \alpha \gamma}{
   \pi \left((\alpha x)^{2} + 1\right) 
   \left((\pi \gamma y)^{2} 
     + \left(2 \arctan{\left(\alpha \tau x \right)} - \pi\right)^{2}\right)^{2} 
   \left((\alpha \tau x)^{2} + 1\right)
} \times  \nonumber \\
\Bigg[
& - 2 (\pi \gamma)^{2} \tau y^{2} 
    \left((\alpha x)^{2} + 1\right) 
    \left(2 \arctan{\left(\alpha x \right)} - \pi\right) \nonumber \\
& + \tau \left((\alpha x)^{2} + 1\right) 
    \left((\pi \gamma y)^{2} 
        + \left(2 \arctan{\left(\alpha \tau x \right)} - \pi\right)^{2}\right) 
    \left(2 \arctan{\left(\alpha x \right)} - \pi\right) \nonumber \\
& - \left((\pi \gamma y)^{2} 
       + \left(2 \arctan{\left(\alpha \tau x \right)} - \pi\right)^{2}\right) 
    \left((\alpha \tau x)^{2} + 1\right) 
    \left(2 \arctan{\left(\alpha \tau x \right)} - \pi\right)
\Bigg] \nonumber
\end{align}
where $\alpha,\gamma > 0 $ and $\tau \in [0, 1]$. Thus the log-likelihood is:
\begin{align}
\ell(\alpha,\gamma,\tau) 
= \sum_{i=1}^n \Bigg\{ 
& \log(4 \alpha \gamma) 
- \log \pi  \nonumber\\
& - \log\!\left((\alpha x_i)^{2} + 1\right) \nonumber \\
& - \log\!\left( (\pi \gamma y_i)^{2} 
      + \left(2 \arctan(\alpha \tau x_i) - \pi\right)^{2} \right)^{2} \nonumber \\
& - \log\!\left((\alpha \tau x_i)^{2} + 1\right) \nonumber\\
& + \log \Bigg[
   - 2 (\pi \gamma)^{2} \tau y_i^{2} 
       \left((\alpha x_i)^{2} + 1\right) 
       \left(2 \arctan(\alpha x_i) - \pi\right) \nonumber\\
& \quad + \tau \left((\alpha x_i)^{2} + 1\right) 
       \left((\pi \gamma y_i)^{2} 
           + \left(2 \arctan(\alpha \tau x_i) - \pi\right)^{2}\right) 
       \left(2 \arctan(\alpha x_i) - \pi\right) \nonumber\\
& \quad - \left((\pi \gamma y_i)^{2} 
           + \left(2 \arctan(\alpha \tau x_i) - \pi\right)^{2}\right) 
       \left((\alpha \tau x_i)^{2} + 1\right) 
       \left(2 \arctan(\alpha \tau x_i) - \pi\right)\nonumber
\Bigg] 
\Bigg\}.
\end{align}

\subsection{Gamma}
We obtain the joint density function by differentiating the joint survival function \eqref{eq: Gamma joint survival}:
\begin{align}
f_{X,Y}(x,y) 
&= \frac{\alpha^{\lambda}\,\gamma^{\nu}\,x^{\lambda-1}\,y^{\nu}}
{\lambda^{1+\tfrac{1}{\nu}}\,y\,\Gamma(\lambda)\,\Gamma(\nu+1)\,
   \Gamma(\lambda,\alpha\tau x)^{2}} \;
e^{-\alpha x - \alpha \tau x 
- \gamma\,\lambda^{-\tfrac{1}{\nu}}\,y\,
   \Gamma(\lambda,\alpha\tau x)^{-\tfrac{1}{\nu}}\,
   \Gamma(\lambda+1)^{\tfrac{1}{\nu}}} \nonumber \\[1em]
&\quad\times
\Biggl[
  \;\gamma \tau^{\lambda} y \,\Gamma(\lambda,\alpha\tau x)^{-\tfrac{1}{\nu}} 
    e^{\alpha x}\,\Gamma(\lambda+1)^{1+\tfrac{1}{\nu}}
    \Gamma(\lambda,\alpha x) \nonumber
\\[0.7em]
&\qquad
 - \lambda^{\tfrac{1}{\nu}}\,\nu \,\tau^{\lambda}\,e^{\alpha x}\,
   \Gamma(\lambda)\,\Gamma(\lambda+1)
 + \lambda^{1+\tfrac{1}{\nu}}\,\nu \,\tau^{\lambda}\,e^{\alpha x}\,
   \Gamma(\lambda)\,\gamma^*(\lambda,\alpha x)
 + \lambda^{1+\tfrac{1}{\nu}}\,\nu\,e^{\alpha\tau x}\,
   \Gamma(\lambda)\,\Gamma(\lambda,\alpha\tau x)
\Biggr], \quad x, y >0, \nonumber
\end{align}

where $\alpha,\gamma > 0 $ and $\tau \in [0, 1]$. Thus the log-likelihood is:
\[
\begin{aligned}
\ell(\alpha,\gamma,\lambda,\nu,\tau)
&= \sum_{i=1}^n \Bigg\{
    \lambda \log \alpha + \nu \log \gamma
   + (\lambda-1)\log x_i + (\nu-1)\log y_i \\[0.4em]
&\quad - \Bigl(1+\tfrac{1}{\nu}\Bigr)\log \lambda
   - \log \Gamma(\lambda) - \log \Gamma(\nu+1)
   - 2 \log \Gamma(\lambda,\alpha\tau x_i) \\[0.4em]
&\quad - \alpha x_i - \alpha \tau x_i
      - \gamma\,\lambda^{-\tfrac{1}{\nu}}\,y_i\,
        \Gamma(\lambda,\alpha\tau x_i)^{-\tfrac{1}{\nu}}\,
        \Gamma(\lambda+1)^{\tfrac{1}{\nu}} \\[0.6em]
&\quad + \log\Bigl[
\;\gamma \tau^{\lambda} y_i \,\Gamma(\lambda,\alpha\tau x_i)^{-\tfrac{1}{\nu}}
   e^{\alpha x_i}\,\Gamma(\lambda+1)^{\,1+\tfrac{1}{\nu}}\,
   \Gamma(\lambda,\alpha x_i) \\[-0.2em]
&\qquad\qquad\;\;
 - \lambda^{\tfrac{1}{\nu}}\,\nu \,\tau^{\lambda}\,e^{\alpha x_i}\,
   \Gamma(\lambda)\,\Gamma(\lambda+1)
 + \lambda^{1+\tfrac{1}{\nu}}\,\nu \,\tau^{\lambda}\,e^{\alpha x_i}\,
   \Gamma(\lambda)\,\gamma^*(\lambda,\alpha x_i) \\
&\qquad\qquad\;\;
 + \lambda^{1+\tfrac{1}{\nu}}\,\nu\,e^{\alpha\tau x_i}\,
   \Gamma(\lambda)\,\Gamma(\lambda,\alpha\tau x_i)
\Bigr]\;\Bigg\}.
\end{aligned}
\]

\section{Method of Moment Estimators} \label{app: mmes}
Suppose data in the form of $\textbf{X}^{(1)}, \textbf{X}^{(2)}, ..., \textbf{X}^{(n)}$ (where for each $i = 1, \ldots, n$, $\textbf{X}^{(i)} = (X_{1i}, X_{2i})^{T} )$ are obtained, which are independent and identically distributed according to \ref{eq: general marginal} and \ref{eq: general cond}. If $M_{1}, M_2 > 0$, consistent asymptotically normal method of moment estimates (m.m.e's) are acquired.\\
Define:
\begin{eqnarray}
	M_{1} &=& \frac{1}{n}\sum_{i = 1}^{n} x_{1i}, \nonumber\\
	M_{2} &=& \frac{1}{n}\sum_{i = 1}^{n} x_{2i}, \nonumber\\
	S_{12}&=& \frac{1}{n}\sum_{i = 1}^{n}(x_{1i} - M_{1} ) ( x_{2i} - M_{2} ),\nonumber \\
    S_{1} &=&  \frac{1}{n}\sum_{i = 1}^{n}(x_{2i} - M_{2} )^2, \nonumber\\
	S_{2} &=&  \frac{1}{n}\sum_{i = 1}^{n}(x_{2i} - M_{2} )^2. \nonumber
\end{eqnarray}
We denote the method of moment estimators for $\alpha, \lambda, \gamma, \nu$ and $\tau$ as $\tilde{\alpha}, \tilde{\lambda}, \tilde{\gamma}, \tilde{\nu}$ and $\tilde{\tau}$ respectively.
\subsection{Exponential}
Using \eqref{eq: exp ex} and \eqref{eq: exp ey}, we obtain:
\begin{align}
    \tilde{\alpha} &= \frac{1}{M_1}, \nonumber\\
    \tilde{\lambda} & = \frac{1}{M_2}. \nonumber
\end{align}
And using \eqref{eq: exp cov}:
\begin{align}
    \tilde{\tau} &= -\frac{1}{\frac{M_1M_2}{S_{12}}+1}. \nonumber
\end{align}
\subsection{Lomax}
Using \eqref{eq: Lomax ex} and \eqref{eq: Lomax varx}, we obtain:
\begin{align}
    \tilde{\lambda} & = \frac{2S_1}{S_1 - M_1^2}, \nonumber \\
    \tilde{\alpha} &= \frac{1}{M_1(\tilde{\lambda}-1 )}. \nonumber
\end{align}
Likewise using \eqref{eq: Lomax ey} and \eqref{eq: Lomax vary}, we obtain:
\begin{align}
    \tilde{\nu} & = \frac{2S_2}{S_2 - M_2^2} ,\nonumber \\
    \tilde{\gamma} &= \frac{1}{M_2(\tilde{\nu}-1 )}. \nonumber
\end{align}
And using \eqref{eq: Lomax cov} we may solve for $\tilde{\tau}$ numerically. Additionally, to ensure existence of positive m.m.e's $\tilde{\lambda}$ and $\tilde{\nu}$, it must be such that $S_1 > M_1^2$ and $S_2 > M_2^2$.
\subsection{Weibull}
Solve for $\tilde{\lambda}$ and $\tilde{\alpha}$ numerically using \eqref{eq: Weibull ex} and $\eqref{eq: Weibull varx}$. Likewise, solve for $\tilde{\nu}$ and $\tilde{\gamma}$ numerically using \eqref{eq: Weibull ey} and $\eqref{eq: Weibull vary}$. Finally, solve for $\tilde{\tau}$ numerically using \eqref{eq: Weibull cov}.
\subsection{Log-Logistic}
Solve for $\tilde{\lambda}$ and $\tilde{\alpha}$ numerically using \eqref{eq: Log ex} and $\eqref{eq: Log varx}$. Likewise, solve for $\tilde{\nu}$ and $\tilde{\gamma}$ numerically using \eqref{eq: Log ey} and $\eqref{eq: Log vary}$. Finally, solve for $\tilde{\tau}$ numerically using \eqref{eq: Log cov}.

\section{Theorems} \label{app: thms}

\begin{theorem} \label{thm: ifr hazard}
Let $h_0:(0,\infty)\to(0,\infty)$ be a non-decreasing (i.e., increasing failure rate) hazard function. Define:
\[
r(x,\tau) = \tau\,\frac{h_0(\tau x)}{h_0(x)}, \qquad x>0.
\]
Then:
\[
0 \le r(x,\tau) \le 1 \quad \forall x>0
\quad \Longleftrightarrow \quad
0 \le \tau \le 1.
\]
\end{theorem}

\begin{proof}
(\emph{Sufficiency.}) Suppose $0 \le \tau \le 1$. Then $\tau x \le x$, and since $h_0$ is non-decreasing,
it follows that $h_0(\tau x) \le h_0(x)$. Consequently,
\[
0 \le r(x,\tau) = \tau\,\frac{h_0(\tau x)}{h_0(x)} \le \tau \le 1,
\]
for all $x>0$.

(\emph{Necessity.}) Conversely, assume that $0 \le r(x,\tau) \le 1$ for all $x>0$. 
Because $h_0(x)>0$, the lower bound $r(x,\tau)\ge 0$ implies $\tau \ge 0$. 
If $\tau>1$, then $\tau x > x$ and, by monotonicity, $h_0(\tau x) \ge h_0(x)$, so:
\[
r(x,\tau) = \tau\,\frac{h_0(\tau x)}{h_0(x)} \ge \tau > 1,
\]
contradicting the upper bound. Hence $\tau \le 1$. 
Combining both parts yields $0 \le \tau \le 1$.
\end{proof}

\begin{theorem}  \label{thm: factor hazard}
Let the hazard function, $h_0:(0,\infty)\to(0,\infty)$, admit the factorization:
\[
h_0(x)=C\,x^{p-1}\,g(x),\qquad C>0,\; p\ge 0,
\]
where $g:(0,\infty)\to(0,\infty)$ is non-decreasing. For $\tau>0$ define:
\[
r(x,\tau)\;=\;\tau\,\frac{h_0(\tau x)}{h_0(x)}.
\]
Then, for every $x>0$,
\[
0\le \tau\le 1
\quad\Longleftrightarrow\quad
0\le r(x,\tau)\le 1.
\]
\end{theorem}
\begin{proof}
Under the stated factorization,
\[
r(x,\tau)=\tau\,\frac{C(\tau x)^{p-1}g(\tau x)}{C x^{p-1}g(x)}
=\tau^{p}\,\frac{g(\tau x)}{g(x)}.
\]
(\emph{Sufficiency.}) If $0\le\tau\le 1$, then $\tau x\le x$ and, since $g$ is non-decreasing,
$g(\tau x)\le g(x)$. Hence:
\[
0\le r(x,\tau)=\tau^{p}\frac{g(\tau x)}{g(x)}\le \tau^{p}\le 1.
\]

(\emph{Necessity.}) Suppose $0\le r(x,\tau)\le 1$ for all $x>0$. Because $h_0>0$ we have $\tau\ge 0$.
If $\tau>1$, then $\tau x>x$ and $g(\tau x)\ge g(x)$, whence:
\[
r(x,\tau)=\tau^{p}\frac{g(\tau x)}{g(x)}\ge \tau^{p}>1,
\]
a contradiction. Thus $\tau\le 1$. Combining, $0\le\tau\le 1$.
\end{proof}
\begin{theorem}\label{thm: half-cauchy}
Let $\alpha>0$. For the half-Cauchy model, the function
\[
\phi(x)\;=\;x\,h_0(x)\;=\;x\,\frac{\alpha}{\bigl(1+(\alpha x)^2\bigr)\bigl(\frac{\pi}{2}-\arctan(\alpha x)\bigr)}
\]
is strictly increasing on $x\in(0,\infty)$.
\end{theorem}

\begin{proof}
It suffices to show $\phi'(x)>0$ for $x>0$. Differentiating and simplifying yields
\[
\phi'(x)
=
\frac{2 \alpha \left(2 \alpha^{2} x^{2} \bigl(2 \arctan(\alpha x) - \pi\bigr) + 2 \alpha x - \left(\alpha^{2} x^{2} + 1\right) \bigl(2 \arctan(\alpha x) - \pi\bigr)\right)}{\left(\alpha^{2} x^{2} + 1\right)^{2} \left(2 \arctan(\alpha x) - \pi\right)^{2}}.
\]
For $x>0$ the denominator is strictly positive because $\left(\alpha^{2} x^{2}+1\right)^2>0$ and $\left(2\arctan(\alpha x)-\pi\right)^2>0$ (since $0<\arctan(\alpha x)<\tfrac{\pi}{2}$, we have $2\arctan(\alpha x)-\pi<0$). Thus the sign of $\phi'(x)$ is determined by
\[
f(x):=2 \alpha^{2} x^{2} \bigl(2 \arctan(\alpha x) - \pi\bigr) + 2 \alpha x - \left(\alpha^{2} x^{2} + 1\right) \bigl(2 \arctan(\alpha x) - \pi\bigr).
\]
Writing $t=\alpha x>0$ gives
\[
f(t)=(t^{2}-1)\bigl(2\arctan t-\pi\bigr)+2t.
\]
If $0<t<1$, then $t^{2}-1<0$ and $2\arctan t-\pi<0$, so their product is positive; adding $2t>0$ implies $f(t)>0$.
If $t=1$, then $f(1)=2>0$.
If $t>1$, use $\arctan t+\arctan(1/t)=\tfrac{\pi}{2}$ to write
\[
f(t)=2t-(t^{2}-1)\bigl(\pi-2\arctan t\bigr)=2t-2(t^{2}-1)\arctan(1/t).
\]
Since $\arctan u<u$ for all $u>0$, with $u=1/t$ we obtain
\[
f(t)>2t-\frac{2(t^{2}-1)}{t}=\frac{2}{t}>0.
\]
Hence $f(t)>0$ for all $t>0$, so $\phi'(x)>0$ for all $x\in(0,\infty)$. Therefore $\phi(x)$ is strictly increasing on $(0,\infty)$.
\end{proof}

\begin{theorem}\label{thm: gamma}
Let $\alpha>0$. For the gamma model, the function
\[
\phi(x)\;=\;x\,h_0(x)\;= x\frac{\alpha(\alpha x)^{\lambda - 1}e^{-\alpha x}}{\Gamma(\lambda, \alpha x)}
\]
is strictly increasing on $x\in(0,\infty)$.
\end{theorem}

\begin{proof}
It suffices to show $\phi'(x)>0$ for $x>0$. Differentiating and simplifying yields
\[
\phi'(x)
=\frac{\alpha\Big(\alpha x(\alpha x)^{2\lambda-2}+(\alpha x)^{\lambda-1}\big(-\alpha x+\lambda\big)e^{\alpha x}\Gamma(\lambda,\alpha x)\Big)e^{-2\alpha x}}
{\Gamma(\lambda,\alpha x)^2}
\]
Let $z=\alpha x>0$. Then we can rewrite as
\[
\phi'(x)
=\frac{\alpha\,z^{\lambda-1}e^{-2z}}{\Gamma(\lambda,z)^2}
\Big[z^{\lambda}+(\lambda-z)e^{z}\Gamma(\lambda,z)\Big].
\]
Since $\alpha>0$, $z^{\lambda-1}>0$, $e^{-2z}>0$, and $\Gamma(\lambda,z)>0$ for all $\lambda,z>0$, it suffices to show
\[
B(z):=z^{\lambda}+(\lambda-z)e^{z}\Gamma(\lambda,z)>0.
\]
Using the recurrence relation for the upper incomplete gamma function,
\[
\Gamma(\lambda+1,z)=\lambda\Gamma(\lambda,z)+z^{\lambda}e^{-z},
\]
we obtain
\[
e^{z}\Gamma(\lambda+1,z)=\lambda e^{z}\Gamma(\lambda,z)+z^{\lambda}.
\]
Hence
\[
B(z)=z^{\lambda}+(\lambda-z)e^{z}\Gamma(\lambda,z)
=e^{z}\Gamma(\lambda+1,z)-z e^{z}\Gamma(\lambda,z)
=e^{z}\big[\Gamma(\lambda+1,z)-z\Gamma(\lambda,z)\big].
\]

\medskip
By definition of the upper incomplete gamma function,
\[
\Gamma(\lambda+1,z)-z\Gamma(\lambda,z)
=\int_{z}^{\infty}\!\!\big(t^{\lambda}-z\,t^{\lambda-1}\big)e^{-t}\,dt
=\int_{z}^{\infty}(t-z)\,t^{\lambda-1}e^{-t}\,dt.
\]
The integrand $(t-z)t^{\lambda-1}e^{-t}$ is nonnegative on $[z,\infty)$ 
and positive almost everywhere on that interval; 
therefore the integral is strictly positive, 
implying $B(z)=e^{z}\big[\Gamma(\lambda+1,z)-z\Gamma(\lambda,z)\big]>0$. Since every factor in the expression for $\phi'(x)$ is positive for $\alpha>0$, $\lambda>0$, and $x>0$, we conclude that $\phi(x)$ is strictly increasing for $x>0$.
\end{proof}

\bibliographystyle{unsrtnat}

\bibliography{main}

\end{document}